\documentclass{llncs}

\usepackage{amssymb,amsmath}
\usepackage[toc,page]{appendix}
\usepackage{placeins}
\usepackage{graphicx}
\usepackage{tikz}
\usetikzlibrary{plotmarks}
\usetikzlibrary{automata} % LTEX and plain TEX
\usetikzlibrary[automata] % ConTEXt
\usetikzlibrary{decorations}
\usepackage{fancybox}
\usetikzlibrary{shapes.arrows,chains}
\usepackage[all]{xy}
\usepackage{bibentry}
\usepackage{graphics}
\usepackage{xspace,aeguill,ae}
\usepackage{verbatim,mathrsfs,enumerate,url}
\usepackage{nag}
\usepackage{wrapfig}

\makeatletter
\renewcommand*{\@fnsymbol}[1]{\ensuremath{\ifcase#1\or \star\or \dagger\or \ddagger\or
       \mathsection\or \mathparagraph\or \|\or **\or \dagger\dagger
       \or \ddagger\ddagger \else\@ctrerr\fi}}
\makeatother

\begin{document}
%-----------------
\newcommand{\setP}{O}
\newcommand{\setH}{H}
\newcommand{\setL}{L}
\newcommand{\setK}{K}
\newcommand{\Sub}[2]{{#1}_{|{#2}}}
\newcommand{\Hist}{Hist}
\newcommand{\Plays}{Plays}
\newcommand{\devstep}[1]{$#1$-deviation step}
\newcommand{\out}[1]{\langle #1 \rangle}
\newcommand{\Succ}{Succ}
\newcommand{\Remove}{\textit{Remove}}
\newcommand{\Adjust}{\textit{Adjust}}
\newcommand{\BotSCC}{{\cal C}}
\newcommand{\Attr}{{Attr}}

%
%-----------------
%
%\newtheorem{definition}{Definition}
%\newtheorem{proposition}[definition]{Proposition}
%\newtheorem{lemma}[definition]{Lemma}
%\newtheorem{theorem}[definition]{Theorem}
%\newtheorem{corollary}[definition]{Corollary}
%\newtheorem{remark}[definition]{Remark}
%\newtheorem{observation}[definition]{Observation}
%\newtheorem{example}[definition]{Example}

\title{On the existence of weak subgame perfect equilibria}
\author{V\'eronique Bruy\`ere\inst{1} \and St\'ephane Le Roux\inst{2}$^{,}$\thanks{Author supported by ERC Starting Grant (279499: inVEST).}$^{,}$\thanks{Le Roux has since moved to TU Darmstadt, Germany.} \and Arno Pauly\inst{2}$^{,\star,}$\thanks{Pauly has since moved to Swansea University, UK.}\and \\ Jean-Fran\c cois Raskin\inst{2}$^{,*}$}

\institute{D\'{e}partement d'informatique, Universit\'{e} de Mons (UMONS), Belgium
\and
D\'{e}partement d'informatique, Universit\'{e} Libre de Bruxelles (ULB), Belgium}

\date{}

\maketitle

\begin{abstract}
We study multi-player turn-based games played on (potentially infinite) directed graphs. An outcome is assigned to every play of the game. Each player has a preference relation on the set of outcomes which allows him to compare plays. We focus on the recently introduced notion of weak subgame perfect equilibrium (weak SPE). This is a variant of the classical notion of SPE, where players who deviate can only use strategies deviating from their initial strategy in a finite number of histories. Having an SPE in a game implies having a weak SPE but the contrary is generally false.

We propose general conditions on the structure of the game graph and on the preference relations of the players that guarantee the existence of a weak SPE, that additionally is finite-memory.
From this general result, we derive two large classes of games for which there always exists a weak SPE: $(i)$ the games with a finite-range outcome function, and $(ii)$ the games with a finite underlying graph and a prefix-independent outcome function. For the second class, we identify conditions on the preference relations that guarantee memoryless strategies for the weak SPE.
\end{abstract}

\section{Introduction}
%-------------------------

Subgame-perfect equilibria (SPEs) are a core solution concept for sequential games. For infinite duration games, they only exist in restricted cases, though. A weaker variant of SPE, \emph{weak SPE} was recently introduced in~\cite{BBMR15}. While an SPE must be resistant to any unilateral deviation of one player, a weak SPE must be resistant to such deviations where the deviating strategy differs from the original one on a \emph{finite number} of histories only, or, equivalently, a \emph{single} history. The latter class of deviating strategies is a well-known notion that for instance appears in the proof of Kuhn's theorem~\cite{kuhn53} with the one-step deviation property.

There are games for which there exists a weak SPE but no SPE~\cite{BBMR15,SV03}. The notion of weak SPE is important for several reasons (more details are given in the related work discussed below). First, for the large class of games with upper-semicontinuous payoff functions and for games played on finite trees, the notions of SPE and weak SPE are equivalent. Second, it is a central technical ingredient used to reason on SPEs as shown in~\cite{BBMR15} and~\cite{Flesch10}. Third, being immune to strategies that finitely deviate from the initial strategy profile may be sufficient from the perspective of synthesis (see more below).

In this paper, we provide the following contributions. First, we identify a general class of games played on potentially infinite graphs and prove that these games always admit weak SPE (Theorem~\ref{thm:generalgraph}). The proof of our result has an algorithmic flavour, and proceeds via transfinite induction. The weak SPEs we construct require only finite memory to execute, meaning that the prescribed action at any history depends only on the current vertex and on the state of some finite automaton. Second, starting from this general existence result, we prove the existence of a weak SPE:
 \begin{itemize}
	\item for games on infinite trees with a \emph{finite} number of outcomes (Theorem~\ref{thm:infinitetree}, reproving a result from \cite{Flesch10});
	\item for games with a \emph{finite} underlying graph and a \emph{prefix-independent} outcome function (Theorem~\ref{thm:finitegraph}).
 \end{itemize}
\noindent
Additionally, in the second result, we identify conditions on the players' outcome preferences that guarantee the existence of a weak SPE composed of \emph{uniform memoryless} strategies only (Theorem~\ref{thm:uniform}).

\paragraph{{\bf Related work}}
The concept of SPE has been first introduced and studied by the game theory community. In~\cite{kuhn53}, Kuhn proves the existence of SPEs in games played on finite trees. This result has been generalized in several ways. All games with a continuous real-valued outcome function and a finitely branching tree always have an SPE~\cite{Roux14} (the special case with finitely many players is first established in~\cite{Fudenberg83}). In~\cite{Flesch10} (resp. \cite{Purves11}), the authors prove that there always exists an SPE for games with a finite number of players and with a real-valued outcome function that is upper-semicontinuous (resp. lower-semicontinuous) and has finite range. The result of \cite{Purves11} is extended to an infinite number of players in~\cite{Flesch17}.
In~\cite{Roux14}, it is proved using Borel determinacy that all two-player games with antagonistic preferences over finitely many outcomes and a Borel-measurable outcome function have an SPE. In~\cite{Roux15}, Le Roux shows that all games where the preferences over finitely many outcomes are free of some ``bad pattern'' and the outcome function is $\Delta^0_2$ measurable (a low level in the Borel hierarchy) have an SPE.

In part of the aforementioned works, the equivalence between SPEs and weak SPEs is implicitly used as a proof technique: in a finite setting in~\cite{kuhn53}, in a continuous setting in \cite{Fudenberg83}, and in a lower-semicontinuous setting in \cite{Flesch10}. In the latter reference, it is implicitly proven that all games with a finite range real-valued outcome function always have a weak SPE (which appears to be an SPE when the outcome function is additionally lower-semicontinuous). We obtain this result here as a consequence of a more general theorem, with a proof of a more algorithmic nature.

The concept of SPE and other solution concepts for multi-player non zero-sum games have been considered recently by the theoretical computer community, see~\cite{BrenguierCHPRRS16} for a survey. The existence of SPEs (and thus weak SPEs) is established in~\cite{Ummels06} for games played on graphs by a finite number of players and with Borel Boolean objectives. In~\cite{BBMR15}, weak SPEs are introduced as a technical tool for showing the existence of SPEs in quantitative reachability games played on finite weighted graphs. An algorithm is also provided for the construction of a (finite-memory) weak SPE that appears to be an SPE for this particular class of games. In this paper, we give several existence results that are orthogonal to the results obtained in \cite{BBMR15} as they are concerned with possibly infinite graphs or prefix-independent outcome functions.

Other refinements of Nash equilibria (NE) are studied. Let us mention the secure equilibria for two players first introduced in~\cite{CHJ06} and then used for reactive synthesis in~\cite{ChatterjeeH07}. These equilibria are generalized to multiple players in~\cite{Depril14} or to quantitative objectives in~\cite{BMR14}, see also a variant called Doomsday equilibrium in \cite{Chatterjee0FR14}. Like NEs, they are subject to possible non-credible threats. Other alternatives to NE are provided by the notion of admissible strategy introduced in~\cite{Berwanger07}, with computational aspects studied in~\cite{BrenguierRS14}, and potential for synthesis studied in~\cite{BrenguierRS15}. Note that these notions are free, like (weak) SPEs, of non-credible threats. Finally, in \cite{KupfermanPV16}, the authors introduce the notion of cooperative and non-cooperative rational synthesis as a general framework where rationality can be specified by either NE, or SPE, or the notion of dominating strategies. In all cases except~\cite{BMR14} and~\cite{Depril14}, the proposed solution concepts are not guaranteed to exist, hence results concern mostly algorithmic techniques to decide their existence, instead of general conditions for existence as in this paper.

\paragraph{{\bf Applications to reactive synthesis}}

Games played on graphs have a large number of applications in theoretical computer science. One particularly important application is \emph{reactive synthesis}~\cite{PR89}, i.e.  the design of a controller that guarantees a good behavior of a reactive system evolving in a possibly hostile environment. One classical model proposed for the synthesis problem is the notion of \emph{two-player zero-sum game played on a graph}. One player is the reactive system and the other one is the environment; the vertices of the graph model their possible states and the edges model their possible actions. Interactions between the players generate an infinite play in the graph which model behaviors of the system within its environment. As one cannot assume cooperation of the environment, the objectives of the two players are considered to be opposite. %, i.e. the environment is adversarial.
Constructing a controller for the system then means devising a \emph{winning strategy} for the player modeling it. Reality is often more subtle and the environment is usually not fully adversarial as it has its own objective, meaning that the game should be non zero-sum. Moreover instead of two players, we could consider the more general situation of several players modeling different interacting systems/environments each of them with its own objective.

This has lead to an exploration of a variety of solution concepts for sequential games from the perspective of theoretical computer science (see survey \cite{bruyere-survey}). Weak SPE have the benefit of allowing less unreasonable threats than Nash equilibria, but existing in more cases than SPE. We can even imagine ruling out infinite deviations by letting a meta-agent punish every one-shot deviation with a (low) fixed probability. A player using an infinitely-deviating strategy will thus be punished by the meta-agent with probability one. Protocols like BitTorrent use similar ideas: every deviant user is temporarily denied suitable bandwidth (see Chapter \textit{Bandwidth Trading as Incentive} in \cite{BitTorrent} for details).

\paragraph{{\bf Structure of the paper}} In Section~\ref{sec:prelim}, we recall the useful notions of game, strategy and weak SPE. In Section~\ref{sec:general}, we present our general conditions that guarantee the existence of a weak SPE. From this general existence result, we derive two large classes of games with a weak SPE: games with a finite-range outcome function in Section~\ref{sec:first}, and games with a finite underlying graph and a prefix-independent outcome function in Section~\ref{sec:second}. In Section \ref{sec:counterexample} we provide an example of a game without weak SPE demonstrating limitations to possible extensions of our main theorem.

An extended abstract omitting most proofs has appeared as \cite{wSPE-fossacs}.

\section{Preliminaries} \label{sec:prelim}
%---------------------------

In this section, we recall the useful notions of game, strategy, and weak subgame perfect equilibrium. We illustrate these notions with examples.

\subsection{Games}
%-----------------------
%
We consider multi-player turn-based games such that an outcome is assigned to every play. Each player has a preference relation on the set of outcomes which allows him to compare plays.

\begin{definition}
A \emph{game} is a tuple $G = (\Pi, V, (V_i)_{i \in \Pi}, E, \setP, \mu, (\prec_i)_{i \in \Pi})$ where:

\begin{itemize}
  \item $\Pi $ is a set of players,
  \item $V$ is a set of vertices and $E \subseteq V \times V$ is a set of edges, such that w.l.o.g. each vertex has at least one outgoing edge,
  \item $(V_i)_{i \in \Pi}$ is a partition of $V$ such that $V_i$ is the set of vertices controlled by player $i \in \Pi$,
  \item $\setP$ is a set of outcomes and $\mu : V^\omega \to \setP$ is an outcome function,
    \item $\prec_i$ $\subseteq \setP \times \setP$ is a preference relation for player $i \in \Pi$.
\end{itemize}
\end{definition}

In this definition the underlying graph $(V,E)$ can be infinite (that is, of arbitrarily cardinality), as well as the set $\Pi$ of players and the set $\setP$ of outcomes.
%The latter set $\setP$ can be any set of abstract payoffs, we will see classical examples in the next section.

A \emph{play} of $G$ is an infinite (countable) sequence $\rho = \rho_0 \rho_1 \ldots \in V^\omega$ of vertices such that $(\rho_i, \rho_{i + 1}) \in E$ for all $i \in \mathbb N$. \emph{Histories} of $G$ are finite
sequences $h = h_0 \ldots h_n \in V^+$ defined in the same way. We often use notation $hv$ to mention the last vertex $v \in V$ of the history. %The \emph{length} $|h|$ of $h$ is the number $n$ of its edges. We denote by $\First(h)$ (resp. $\Last(h)$) the first vertex $h_0$ (resp. last vertex $h_n$) of $h$.
Usually histories are non empty, but in specific situations it will be useful to consider the empty history $\epsilon$. The set of plays is denoted by $\Plays$ and the set of histories (ending with a vertex in $V_i$) by $\Hist$ (resp. by $\Hist_i$).\footnote{Indexing $\Plays_G$ or $\Hist_G$ with $G$ allows to recall the related game $G$.}
A \emph{prefix} (resp. \emph{suffix}) of a play $\rho = \rho_0\rho_1 \ldots$ is a finite sequence $\rho_{\leq n} = \rho_0 \dots \rho_n$  (resp. infinite sequence $\rho_{\geq n} = \rho_n \rho_{n+1} \ldots$). We use notation $h < \rho$ when a history $h$ is prefix of a play $\rho$.
When an initial vertex $v_0 \in V$ is fixed, we call $(G, v_0)$ an \emph{initialized} game. In this case, plays and histories are supposed to start in $v_0$, and we use notations $\Plays(v_0)$ and $\Hist(v_0)$. In this article, we often \emph{unravel} the graph of the game $(G, v_0)$ from the initial vertex $v_0$, which yields an infinite tree rooted at $v_0$.

The outcome function assigns an outcome $\mu(\rho) \in \setP$ to each play $\rho \in V^\omega$. It is \emph{prefix-independent} if $\mu(h \rho) = \mu(\rho)$ for all histories $h$ and play $\rho$. A \emph{preference} relation $\prec_i$ $\subseteq \setP \times \setP$ is an irreflexive and transitive binary relation.
%Acyclic transitive relation $\prec_i$ ? Notation $a \nprec_i b$ instead of $b \preceq_i a$}.
It allows for player~$i$ to compare two plays $\rho, \rho' \in V^\omega$ with respect to their outcome: $\mu(\rho) \prec_i \mu(\rho')$ means that player~$i$ prefers $\rho'$ to $\rho$. In this paper we restrict to \emph{linear} preferences. (It is w.l.o.g. since the preference properties that we use are preserved by linear extension).
We write $o \preceq_i o'$ when $o \prec_i o'$ or $o = o'$; notice that $o \nprec_i o'$ if and only if $o' \preceq_i o$.  We sometimes use notation $\prec_v$ instead of $\prec_i$ when vertex $v \in V_i$ is controlled by player~$i$.

\begin{example} \label{ex:classical}
Let us mention some classical classes of games where the set of outcomes~$\setP$ is a subset of $(\mathbb R \cup \{+\infty, -\infty \})^\Pi$, and for all player~$i \in \Pi$, $\prec_i$ is the usual ordering $<$ on $\mathbb R \cup \{+\infty, -\infty \}$ on the outcome $i$-th components. In other words, each player~$i$ has a real-valued payoff function $\mu_i : \Plays \to \mathbb R \cup \{+\infty, -\infty \}$. The outcome function of the game is then equal to $\mu = (\mu_i)_{i \in \Pi}$, and for all $i \in \Pi$, $\mu(\rho) \prec_i \mu(\rho')$ whenever $\mu_i(\rho) < \mu_i(\rho')$.

Games with \emph{Boolean} objectives are such that $\mu_i : \Plays \to \{0,1\}$ where $1$ (resp. $0$) means that the play is won (resp. lost) by player~$i$. Classical objectives are Borel objectives including $\omega$-regular objectives, like reachability, B\"uchi, parity, aso~\cite{GU08}. Prefix-independence of $\mu_i$ holds in the case of B\"uchi and parity objectives, but not for reachability objective.

We have \emph{quantitative} objectives when $\mu_i : \Plays \to \mathbb R \cup \{+\infty, -\infty \}$ replaces $\mu_i : \Plays \to \{0,1\}$. Usually, such a $\mu_i$ is defined from a weight function $w_i : E \to \mathbb R$ that assigns a weight to each edge. Classical examples of $\mu_i$ are \emph{limsup} and \emph{mean-payoff} functions~\cite{LaurentDoyen}, that is\footnote{The limit inferior can be used instead of the limit superior.},
\begin{itemize}
\item \emph{limsup}: $\mu_i(\rho) = \limsup_{k \to \infty} w_i(\rho_k,\rho_{k+1})$
\item \emph{mean-payoff}: $\mu_i(\rho) = \limsup_{n \to \infty} \sum_{k=0}^{n} \frac{w_i(\rho_k,\rho_{k+1})}{n}$
\end{itemize}
%the limit superior of the weights seen along a play $\rho$ ($\mu_i(\rho) = \limsup_{k \to \infty} w_i(\rho_k,\rho_{k+1})$),  the limit superior of the mean of the weights seen along a play ($\mu_i(\rho) = \limsup_{n \to \infty} \sum_{k=0}^{n} \frac{w_i(\rho_k,\rho_{k+1})}{n}$) \cite{LaurentDoyen}.
\end{example}

\subsection{Strategies}% and deviations}
%----------------------------------------------

Let $(G, v_0)$ be an initialized game. A \emph{strategy} $\sigma$ for player~$i$ in $(G,v_0)$ is a function $\sigma: \Hist_i(v_0) \to V$ assigning to each history $hv \in \Hist_i(v_0)$ a vertex $v' = \sigma(hv)$ such that $(v, v') \in E$.
%In an initialized game $(\mathcal{G}, v_0)$, $\sigma$ is restricted to histories in $\Hist_i(v_0)$.
A strategy $\sigma$ of player $i$ is \emph{positional} if it only depends on the last vertex of the history, i.e. $\sigma(hv) = \sigma(v)$ for all $hv \in \Hist_i(v_0)$. It is a \emph{finite-memory} strategy if it can be encoded by a deterministic \emph{Moore machine} ${\cal M} = (M, m_0, \alpha_U, \alpha_N)$ where $M$ is a finite set of states (the memory of the strategy), $m_0 \in M$ is an initial memory state, $\alpha_U : M \times V \rightarrow M$ is an update function, and $\alpha_N : M \times V_i \rightarrow V$ is a next-move function.\footnote{Moore machines are usually defined for finite sets $V$ of vertices. We here allow infinite sets $V$.} Such a machine defines a strategy $\sigma$ such that $\sigma(hv) = \alpha_N(\widehat{\alpha}_U(m_0,h),v)$ for all histories $hv \in \Hist_i(v_0)$, where $\widehat{\alpha}_U$ extends $\alpha_U$ to histories as expected. The \emph{memory size} of $\sigma$ is then the size $|M|$ of $\cal M$. In particular $\sigma$ is positional when it has memory size one.

The previous definitions of (positional, finite-memory) strategy are given for an initialized game $(G, v_0)$. We call \emph{uniform} every positional strategy $\sigma$ of player~$i$ defined for all $hv \in \Hist_i$ (instead of $\Hist_i(v_0)$), that is, when $\sigma$ is a positional strategy in all initialized games $(G,v)$, $v \in V$.

A play $\rho$ is \emph{consistent} with a strategy $\sigma$ of player~$i$ if $\rho_{n+1} = \sigma(\rho_{\leq n})$ for all $n$ such that $\rho_n \in V_i$. A \emph{strategy profile} is a tuple $\bar\sigma = (\sigma_i)_{i \in \Pi}$ of strategies, where each $\sigma_i$ is a strategy of player~$i$. It is called \emph{positional} (resp. \emph{finite-memory with memory size bounded by $c$}, \emph{uniform}) if all $\sigma_i$, $i \in \Pi$, are positional (resp. finite-memory with memory size bounded by $c$, uniform).
%\textcolor{blue}{When there exists a constant $c$ such that the memory size of all $\sigma_i$, $i \in \Pi$, is bounded by $c$, we say that $\bar\sigma$ has a memory size bounded by $c$.}
Given an initial vertex $v_0$, such a strategy profile determines a unique play of $(G, v_0)$ that is consistent with all the strategies. This play induced by $\bar\sigma$ in $(G,v_0)$ is denoted by $\out{\bar\sigma}_{v_0}$ and we say that $\bar\sigma$ has outcome $\mu(\out{\bar\sigma}_{v_0})$.

Let $\bar \sigma$ be a strategy profile. When all players stick to their own strategy except player $i$ that shifts from $\sigma_i$ to $\sigma'_i$, we denote by $(\sigma'_i, \bar \sigma_{-i})$ the derived strategy profile, and by $\out{\sigma'_i, \bar \sigma_{-i}}_{v_0}$ the induced play in $(G, v_0)$. We say that $\sigma'_i$ is a \emph{deviating} strategy from $\sigma_i$. When $\sigma_i$ and $\sigma'_i$ only differ on a finite number of histories (resp. on $v_0$), we say that $\sigma'_i$ is a \emph{finitely-deviating} (resp. \emph{one-shot deviating}) strategy from $\sigma_i$. One-shot deviating strategies is a well-known notion that for instance appears in the proof of Kuhn's theorem~\cite{kuhn53} with the one-step deviation property. Finitely-deviating strategies have been introduced in \cite{BBMR15}.

\subsection{Variants of subgame perfect equilibria}
%------------------------------------------------------------------

In this section we recall the notion of subgame perfect equilibrium (SPE) and its variants. Let us first recall the classical notion of Nash equilibrium (NE). Informally, a strategy profile $\bar\sigma$ in an initialized game $(G,v_0)$ is an NE if no player has an incentive to deviate (with respect to his preference relation), if the other players stick to their strategies.

\begin{definition}
Given an initialized game $(G, v_0)$, a strategy profile $\bar \sigma = (\sigma_i)_{i \in \Pi}$ of $(ùG,v_0)$ is a \emph{Nash equilibrium} if for all players $i \in \Pi$, for all strategies $\sigma'_i$ of player~$i$, we have
$\mu(\out{\bar \sigma}_{v_0}) \nprec_i \mu(\out{\sigma'_i, \bar \sigma_{-i}}_{v_0})$.
\end{definition}

When $\mu(\out{\bar \sigma}_{v_0}) \prec_i \mu(\out{\sigma'_i, \bar \sigma_{-i}}_{v_0})$, we say that $\sigma'_i$ is a \emph{profitable deviation} for player~$i$ w.r.t. $\bar\sigma$.

The notion of subgame perfect equilibrium is a refinement of NE. In order to define it, we need to introduce the following concepts. Given a game $G = (\Pi, V, (V_i)_{i \in \Pi}, E, \mu, (\prec_i)_{i \in \Pi})$ and a history $h \in \Hist$, we denote by $\Sub{G}{h}$ the game $(\Pi, V, (V_i)_{i \in \Pi}, E, \Sub{\mu}{h}, (\prec_i)_{i \in \Pi})$ where $\Sub{\mu}{h}(\rho) = \mu(h\rho)$ for all plays of $\Sub{G}{h}$\footnote{In this article, we will always use notation $\mu(h\rho)$ instead of $\Sub{\mu}{h}(\rho)$.}, and we say that $\Sub{G}{h}$ is a \emph{subgame} of $G$. Given an initialized game $(G, v_0)$ and a history $hv \in \Hist(v_0)$, the initialized game $(\Sub{G}{h}, v)$ is called the subgame of $(G, v_0)$ with history $hv$. In particular $(G, v_0)$ is a subgame of itself with history $hv_0$ such that $h = \epsilon$. Given a strategy $\sigma$ of player $i$ in $(G,v_0)$, the strategy $\Sub{\sigma}{h}$ in $(\Sub{G}{h}, v)$ is defined as $\Sub{\sigma}{h}(h') = \sigma(hh')$ for all histories $h' \in \Hist_i(v)$. Given a strategy profile $\bar \sigma$ in $(G,v_0)$, we use notation $\Sub{\bar \sigma}{h}$ for $(\Sub{\sigma_i}{h})_{i \in \Pi}$, and $\out{\Sub{\bar \sigma}{h}}_{v}$ is the play induced by $\Sub{\bar \sigma}{h}$ in the subgame $(\Sub{G}{h}, v)$.

We can now recall the classical notion of subgame perfect equilibrium: an SPE is a strategy profile in an initialized game that induces an NE in each of its subgames.
%In particular, a subgame perfect equilibrium is an NE.
Two variants of SPE, called weak SPE and very weak SPE, are proposed in~\cite{BBMR15} such that no player has an incentive to deviate in any subgame using finitely deviating strategies and one-shot deviating strategies respectively (instead of any deviating strategy). %Very weak SPE are also introduced \textcolor{red}{in~\cite{???} reference required} under the terminology of one-deviation immune SPE.

\begin{definition}
%Given an initialized game $(G, v_0)$, a strategy profile $\bar \sigma$ of $(G,v_0)$ is a \emph{(weak, very weak) subgame perfect equilibrium (SPE)} if $\Sub{\bar \sigma}{h}$ is a (weak, very weak) NE in $(\Sub{G}{h}, v)$, for all histories $hv \in \Hist(v_0)$.
Given an initialized game $(G, v_0)$, a strategy profile $\bar \sigma$ of $(G,v_0)$ is a \emph{(weak, very weak resp.) subgame perfect equilibrium} if for all histories $hv \in \Hist(v_0)$, for all players $i \in \Pi$, for all (finitely, one-shot resp.) deviating strategies $\sigma'_i$ from $\Sub{\sigma_i}{h}$ of player~$i$ in the subgame $(\Sub{G}{h}, v)$, we have
$\mu(\out{\Sub{\bar \sigma}{h}}_{v}) \nprec_i \mu(\out{\sigma'_i, \bar\sigma_{-i | h}}_{v})$.
\end{definition}

Trivially, every SPE is a weak SPE, and every weak SPE is a very weak SPE.

\begin{proposition}[\cite{BBMR15}] \label{prop:weak-veryweak}
Let  $\bar \sigma$ be a strategy profile in $(G, v_0)$. Then $\bar \sigma$ is a weak SPE iff $\bar \sigma$ is a very weak SPE.
There exists an initialized game $(G, v_0)$ with a weak SPE but no SPE.
\end{proposition}

%\begin{example}
%We come back to the game depicted in Figure~\ref{TwoPlayerQuantitativeGame}. We have seen before that the strategy profile $(\sigma_1, \sigma_2)$ is a NE. However it is not an SPE. Indeed consider the subgame $(\SB{\G}{v_0}, v_2)$ of $(\G,v_0)$ with history $v_0v_2$. In this subgame, $\sigma'_2$ is a profitable deviation for player~$2$. One can easily verify that the strategy profile $(\sigma'_1, \sigma'_2)$ is an SPE, as well as a weak SPE and a very weak SPE, due to the simple form of the game.
%\end{example}
%
%The previous example is too simple to show the differences between classical SPEs and their variants.

%The next example presents a game with a (very) weak SPE but no SPE.

\begin{figure}[ht!]
\begin{center}

\begin{tikzpicture}[initial text=,auto, node distance=2cm, shorten >=1pt] %,on grid,auto]

\node[state, scale=0.6]              (1)                     {$v_0$};
\node[state, rectangle, scale=0.6]   (2)    [right=of 1]     {$v_1$};
\node[state, scale=0.6]              (3)    [left=of 1]      {$v_2$};
\node[state, scale=0.6]              (4)    [right=of 2]     {$v_3$};

\path[->] (1) edge [bend right=25, thick, black]          node[below, scale=0.7, black]    {}   (2)
             edge                                node[above, scale=0.7]           {}   (3)

         (2) edge  [bend left=-25]               node[above, scale=0.7]           {}   (1)
             edge  [thick, black]                        node[above, scale=0.7, black]    {}   (4)

         (3) edge  [loop above, thick, black]            node[midway, scale=0.7, black]   {}   ()

         (4) edge  [loop above, thick, black]            node[midway, scale=0.7, black]   {}   ();

\end{tikzpicture}

%\begin{tikzpicture}[initial text=,auto, node distance=2cm, shorten >=1pt] %,on grid,auto]
%
%\node[state, scale=0.5]              (1)                     {$v_0$};
%\node[state, rectangle, scale=0.5]   (2)    [right=of 1]     {$v_1$};
%\node[state, scale=0.5]              (3)    [left=of 1]      {$v_2$};
%\node[state, scale=0.5]              (4)    [right=of 2]     {$v_3$};
%
%\node (fictif) [above left=4 mm of 1]  {};
%
%\path[->] (1) edge [bend right=25, thick, black]          node[below, scale=0.6, black]    {}   (2)
%              edge                                node[above, scale=0.6]           {}   (3)
%
%          (2) edge  [bend left=-25]               node[above, scale=0.6]           {}   (1)
%              edge  [thick, black]                        node[above, scale=0.6, black]    {}   (4)
%
%
%          (3) edge  [loop above, thick, black]            node[midway, scale=0.6, black]   {}   ()
%
%
%          (4) edge  [loop above, thick, black]            node[midway, scale=0.6, black]   {}   ()
%
%          (fictif)   edge (1);
%
%\end{tikzpicture}
\end{center}
\caption{A initialized game $(G,v_0)$ with a (very) weak SPE and no SPE.}
\label{fig:gameNoSPE}
\end{figure}
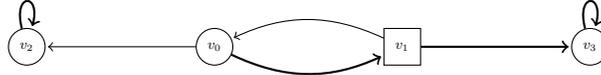

\begin{example}[\cite{BBMR15}] \label{ex:contrex}
Consider the two-player game $(G, v_0)$ in Figure~\ref{fig:gameNoSPE} such that player~$1$ (resp. player~$2$) controls vertices $v_0, v_2, v_3$ (resp. vertex $v_1$). The set $\setP$ of outcomes is equal to $\{o_1, o_2, o_3 \}$, and the outcome function is prefix-independent such that $\mu((v_0v_1)^\omega) = o_1$, $\mu(v_2^\omega) = o_2$, and $\mu(v_3^\omega) = o_3$. The preference relation for player~$1$ (resp. player~$2$) is $o_1 \prec_1 o_2 \prec_1 o_3$ (resp. $o_2 \prec_2 o_3 \prec_2 o_1$).

It is known that this game has no SPE~\cite{SV03}. Nevertheless the positional strategy profile $\bar \sigma$ depicted with thick edges is a very weak SPE, and thus a weak SPE by Proposition~\ref{prop:weak-veryweak}. Let us give some explanation. Due to the simple form of the game, only two cases are to be treated. Consider first the subgame $(\Sub{G}{h}, v_0)$ with $h \in (v_0v_1)^\ast$, and the one-shot deviating strategy $\sigma'_1$ from $\Sub{\sigma_1}{h}$ such that $\sigma'_1(v_0) = v_2$. Then $\out{\Sub{\bar\sigma}{h}}_{v_0} = v_0v_1v_3^\omega$ and $\out{\sigma'_1,\Sub{\sigma_2}{h}}_{v_0} = v_0v_2^\omega$ with respective outcomes $o_3$ and $o_2$, showing that $\sigma'_1$ is not a profitable deviation for player~$1$ in $(\Sub{G}{h}, v_0)$. Now in the subgame $(\Sub{G}{h}, v_1)$ with $h \in (v_0v_1)^\ast v_0$, the one-shot deviating strategy from $\Sub{\sigma_2}{h}$ such that $\sigma'_2(v_1) = v_0$ is not profitable for player~$2$ in $(\Sub{G}{h}, v_1)$ because $\out{\Sub{\bar\sigma}{h}}_{v_1} = v_1v_3^\omega$ and $\out{\Sub{\sigma_1}{h},\sigma'_2}_{v_1} = v_1v_0v_1v_3^\omega$ with the same outcome $o_3$.

Notice that $\bar \sigma$ is not an SPE. Indeed the strategy $\sigma'_2$ such that $\sigma'_2(hv_1) = v_0$ for all $h$, is infinitely deviating from $\sigma_2$, and is a profitable deviation for player~$2$ in $(G, v_0)$ since $\out{\sigma_1,\sigma'_2}_{v_0} = (v_0v_1)^\omega$ with outcome~$o_1$.
\end{example}

\section{General conditions for the existence of weak SPEs} \label{sec:general}
%-----------------------------------------------------------------------------

In this section, we propose general conditions to guarantee the existence of weak SPEs. In the next sections, from this result, we will derive two interesting large families of games always having a weak SPE.

\begin{theorem} \label{thm:generalgraph}
Let $(G,v_0)$ be an initialized game with a subset $\setL \subseteq V$ of vertices called \emph{leaves} with only one outgoing edge $(l,l)$ for all $l \in \setL$. Suppose that:
\begin{enumerate}
\item for all $v \in V$, there exists a play $\rho = hl^{\omega}$ for some $h \in \Hist(v)$ and $l \in \setL$,
\item for all plays $\rho = hl^\omega$ with $h \in \Hist(v)$ and $l \in \setL$, $\mu(\rho) = \mu(l^\omega)$,
%\item for all plays $\rho = hl^\omega$ with $h \in \Hist(v_0)$ and $l \in \setL$, for all suffixes $\rho'$ of $\rho$, $\mu(\rho) = \mu(\rho') = \mu(l^\omega)$,
\item the set of outcomes $\setP_\setL = \{\mu(l^{\omega}) \mid l \in \setL\}$ is finite.
\end{enumerate}
%\begin{enumerate}
%\item there exists a non empty subset $\setL \subseteq V$ of vertices called \emph{leaves} with only one outgoing edge $(l,l)$ for all $l \in \setL$,
%\item for all plays $\rho = hl^\omega$ with $h \in \Hist(v_0)$ and $l \in \setL$, for all suffixes $\rho'$ of $\rho$, $\mu(\rho) = \mu(\rho') = \mu(l^\omega)$,
%\item for all $v \in V$, there exists a play $\rho = hl^{\omega}$ for some $h \in \Hist(v)$ and $l \in \setL$,
%\item the set of payoffs $\setP_\setL = \{\mu(l^{\omega}) \mid l \in \setL\}$ is finite.
%\end{enumerate}
Then there always exists a weak SPE $\bar \sigma$ in $(G, v_0)$. Moreover, $\bar \sigma$ is finite-memory with memory size bounded by $|\setP_\setL|$.
%\fbox{I must add that the outcome of $\sigma$ of each subgame ends with $l^{\omega}$ with $l \in L$ as it is useful for the two corollaries.}
\end{theorem}

Let us comment the hypotheses. The first condition means that from each vertex $v$ of the game there is a leaf reachable from $v$; in particular $L$ is not empty. The second condition expresses prefix-independence of the outcome function restricted to plays %of $\Plays(v_0)$
eventually looping in a leaf $l \in \setL$. The last condition means that even if there is an infinite number of leaves, the set of outcomes assigned by $\mu$ to plays eventually looping in $\setL$ is finite. The next example describes a family of games satisfying the conditions of Theorem~\ref{thm:generalgraph}.

\begin{example} \label{ex:Gn}
For each natural number $n \geq 3$, we build a game $G_n$ with $n$ players, $2n$ vertices, $3n$ edges, and $n+1$ outcomes. The set of players is $\Pi = \{1,2, \ldots, n\}$ and the set of vertices is $V = \{v_1,\ldots,v_n, l_1, \ldots l_n \}$ such that $V_i = \{v_i, l_i\}$ for all $i \in \Pi$. The edges are $(v_1, v_2), (v_2,v_3), \ldots, (v_n,v_1)$, and $(v_i, l_i), (l_i,l_i)$ for all $i \in \Pi$. The game $G_4$ is depicted in Figure~\ref{fig:G4}. The set $\setP$ of outcomes is equal to $\{o_1, \ldots, o_n, \bot \}$, and the outcome function is prefix-independent such that $\mu((v_1v_2 \ldots v_n)^\omega) = \bot$ and $\mu(l_i^\omega) = o_i$ for all $i \in \Pi$. Each player $i$ has a preference relation $\prec_i$ satisfying $\bot \prec_i o_{i-1} \prec_i o_i \prec_i o_j$ for all $j \in \Pi \setminus \{i-1,i\}$ (with the convention that $o_0 = o_n$).

\begin{figure}
\centering

\begin{tikzpicture}[shorten >=1pt,node distance=2cm, auto]
 \node[state, scale=0.6] (a1) {$v_1$};
 \node[state, scale=0.6] (a2) [right of = a1] {$v_2$};
 \node[state, scale=0.6] (a3) [below of = a2] {$v_3$};
 \node[state, scale=0.6] (a4) [left of = a3] {$v_4$};
 \node[state, scale=0.6] (x1) [above left of = a1]{$l_1$};
 \node[state, scale=0.6] (x2) [above right of = a2]{$l_2$};
 \node[state, scale=0.6] (x3) [below right of = a3]{$l_3$};
 \node[state, scale=0.6] (x4) [below left of = a4]{$l_4$};

\path[->] (a1) edge node {} (a2)
		(a2) edge node {} (a3)
		(a3) edge node {} (a4)		
		(a4) edge node {} (a1)
		(a1) edge node {} (x1)		
		(a2) edge node {} (x2)
		(a3) edge node {} (x3)		
		(a4) edge node {} (x4)
		(x1) edge [loop above] node {} ()
		(x2) edge [loop above] node {} ()
		(x3) edge [loop below] node {} ()
		(x4) edge [loop below] node {} ();
\end{tikzpicture}
 \caption{Game $G_4$}
 \label{fig:G4}
 \end{figure}
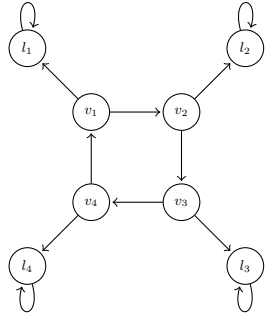

Each game $(G_n,v_1)$ satisfies the hypotheses of Theorem~\ref{thm:generalgraph} with $\setL = \{l_1,\ldots,l_n\}$ and thus has a finite-memory weak SPE. Such a strategy profile $\bar \sigma$ is depicted in Figure~\ref{fig:profileG4} for $n = 4$ (see the thick edges on the unravelling of $G_4$ from the initial vertex $v_1$) and can be easily generalized to every $n \geq 3$. One verifies that this profile is a very weak SPE, and thus a weak SPE by Proposition~\ref{prop:weak-veryweak}. For all $i \in \Pi$, the strategy $\sigma_i$ of player~$i$ is finite-memory with a memory size equal to $n-1$. Intuitively, along $(v_1 \ldots v_n)^\omega$, player~$i$ repeatedly produces one move $(v_i,l_i)$ followed by $n-2$ moves $(v_i,v_{i+1})$. Hence the memory states of the Moore machine for $\sigma_i$ are counters from $1$ to $n-1$. The Moore machine for $\sigma_1$ in the game $(G_4,v_1)$ is depicted in Figure~\ref{fig:MooreG4} (with $M = \{1,2,3\}$, $m_0 = 1$, and the update and next-move functions indicated by the edges).

\begin{figure}
\centering
\begin{tikzpicture}[shorten >=1pt,node distance=.9cm, auto]
  \node (a11) {$v_1$};
  \node (a21) [right of = a11] {$v_2$};
  \node (a31) [right of = a21] {$v_3$};
  \node (a41) [right of = a31] {$v_4$};
  \node (a12) [right of = a41] {$v_1$};
  \node (a22) [right of = a12] {$v_2$};
  \node (a32) [right of = a22] {$v_3$};
  \node (a42) [right of = a32] {$v_4$};
  \node (a13) [right of = a42] {$v_1$};
  \node (a23) [right of = a13] {$v_2$};
  \node (a33) [right of = a23] {$v_3$};
  \node (a43) [right of = a33] {$v_4$};
  \node (a14) [right of = a43] {$v_1$};
  \node (inf) [right of = a14] {};

  \node (x11) [below of = a11]{$l_1$};
  \node (x21) [below of = a21] {$l_2$};
  \node (x31) [below of = a31] {$l_3$};
  \node (x41) [below of = a41] {$l_4$};
  \node (x12) [below of = a12] {$l_1$};
  \node (x22) [below of = a22] {$l_2$};
  \node (x32) [below of = a32] {$l_3$};
  \node (x42) [below of = a42] {$l_4$};
  \node (x13) [below of = a13] {$l_1$};
  \node (x23) [below of = a23] {$l_2$};
  \node (x33) [below of = a33] {$l_3$};
  \node (x43) [below of = a43] {$l_4$};
  \node (x14) [below of = a14] {$l_1$};

\draw[->] (a11) edge node {} (a21);
\draw[->] (a21) edge [thick] node {} (a31);
\draw[->] (a31) edge [thick] node {} (a41);	
\draw[->] (a41) edge node {} (a12);
\draw[->] (a12) edge [thick] node {} (a22)	;	
\draw[->] (a22) edge [thick] node {} (a32);
\draw[->] (a32) edge node {} (a42);	
\draw[->] (a42) edge  [thick]node {} (a13);
\draw[->] (a13) edge  [thick]node {} (a23);
\draw[->] (a23) edge node {} (a33);		
\draw[->] (a33) edge  [thick]node {} (a43);
\draw[->] (a43) edge  [thick] node {} (a14);
		
\draw[dashed] (a14) edge node {} (inf);

\draw[->] (a11) edge [thick] node {} (x11);
\draw[->] (a21) edge node {} (x21);
\draw[->] (a31) edge node {} (x31);
\draw[->] (a41) edge [thick] node {} (x41);
\draw[->] (a12) edge node {} (x12);
\draw[->] (a22) edge node {} (x22);
\draw[->] (a32) edge [thick] node {} (x32);
\draw[->] (a42) edge node {} (x42);
\draw[->] (a13) edge node {} (x13);
\draw[->] (a23) edge [thick] node {} (x23);
\draw[->] (a33) edge node {} (x33);
\draw[->] (a43) edge node {} (x43);
\draw[->] (a14) edge [thick] node {} (x14);
\end{tikzpicture}
\caption{Weak SPE in $(G_4,v_1)$}
\label{fig:profileG4}
\end{figure}
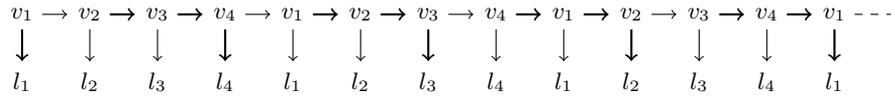

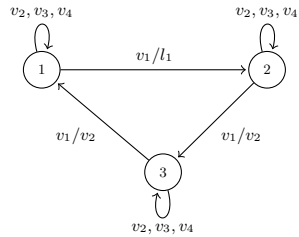
\begin{figure}
\centering

\begin{tikzpicture}[initial text=, auto, shorten >=1pt,node distance=2cm, auto]
 \node[state, scale=0.6] (1) {$1$};
 \node[state, scale=0.6] (2) [right=2.5cm of 1] {$2$};
 \node (fictif) [right=1.25cm of 1] {};
 \node[state, scale=0.6] (3) [below=1cm of fictif] {$3$};

\path[->] (1) edge node [scale=0.7] {$v_1/l_1$} (2)
		  (2) edge node [scale=0.7] {$v_1/v_2$} (3)
		  (3) edge node [scale=0.7] {$v_1/v_2$} (1)
		  (1) edge [loop above] node [scale=0.7] {$v_2,v_3,v_4$} ()
		  (2) edge [loop above] node [scale=0.7] {$v_2,v_3,v_4$} ()
		  (3) edge [loop below] node [scale=0.7] {$v_2,v_3,v_4$} ();
\end{tikzpicture}
\caption{The Moore machine for $\sigma_1$}
\label{fig:MooreG4}
\end{figure}

\end{example}

Let us now proceed to the proof of Theorem~\ref{thm:generalgraph}. Recall that it is enough to prove the existence of a very weak SPE by Proposition~\ref{prop:weak-veryweak}. The proof idea is the following one. Initially, for each vertex $v$, we accept all plays $\rho = hl^{\omega}$ with $h \in \Hist(v)$ and $l \in \setL$ as \emph{potential} plays induced by a very weak SPE in the initialized game $(G,v)$. We thus label each $v$ by the set of outcomes $\mu(l^\omega)$ for such leaves $l$ (recall that $\mu(\rho) = \mu(l^\omega)$ by the second condition of Theorem~\ref{thm:generalgraph}). Notice that this labeling is finite (resp. not empty) by the third (resp. first) condition of the theorem. Step after step, we are going to remove some outcomes from the vertex labelings by a \emph{\Remove} operation followed by an \emph{\Adjust} operation. The \Remove\ operation removes an outcome $o$ from the labeling of a given vertex $v$ when there exists an edge $(v,v')$ for which $o \prec_v o'$ for all outcomes $o'$ that label $v'$. Indeed $o$ cannot be the outcome of a play induced by a very weak SPE since the player who controls $v$ will choose the move $(v,v')$ to get a preferable outcome $o'$. Now it may happen that for another vertex $u$ having $o$ in its labeling, all potential plays induced by a very weak SPE from $u$ with outcome $o$ necessarily cross vertex $v$.  As $o$ has been removed from the labeling of $v$, these potential plays do no longer survive and $o$ will also be removed from the labeling of $u$ by the \Adjust\ operation. Repeatedly applying these two operations converge to a fixpoint for which we will prove non-emptiness (this is the difficult part of the proof, non-emptiness will be obtained by maintaining three invariants, see Lemma~\ref{lem:fixpoint}).
%From the resulting labeling of the vertices, we will show how to build a very weak SPE in $(G,v_0)$.
From this fixpoint, for each vertex $v$ and each outcome $o$ of the resulting labeling of $v$, there exists a play $\rho_{v,o}= hl^{\omega}$ with outcome $o$ for some $h \in \Hist(v)$ and $l \in \setL$. We can thus build a very weak SPE $\bar \sigma$ in $(G,v_0)$ as follows.  The construction of $\bar \sigma$ is done step by step: \emph{(i)} initially $\bar \sigma$ is partially defined such that $\out{\bar \sigma}_{v_0} = \rho_{v_0,o_0}$ for some $o_0$; \emph{(ii)} then in the subgame $(\Sub{G}{h},v)$ such that $\out{\Sub{\bar \sigma}{h}}_{v} = \rho_{v,o}$, if the player who controls $v$ chooses the move $(v,v')$ in a one-shot deviation, then there exists $\rho_{v',o'}$ such that $o \nprec_v o'$ by definition of the fixpoint, and we thus extend the construction of $\bar \sigma$ such that $\out{\Sub{\bar \sigma}{hv}}_{v'} = \rho_{v',o'}$.

Let us now go into the details of the proof. For each $l \in \setL$, we denote by $o_l$ the outcome $\mu(l^{\omega})$. Recall that for all $\rho = hl^{\omega}$ we have $\mu(\rho) = o_l$ by the second hypothesis of the theorem. For each $v \in V$, we denote by $\Succ(v)$ the set of successors of $v$ distinct from $v$, that is, the vertices $v' \neq v$ such that $(v,v') \in E$. Notice that the leaves $l$ are the vertices with only one outgoing edge $(l,l)$. Thus, by definition, $\Succ(v) = \emptyset$ for all $v \in \setL$ and $\Succ(v) \neq \emptyset$ for all $v \in V \setminus \setL$.

The labeling $\lambda_{\alpha}(v)$ of the vertices $v$ of $G$ by subsets of $\setP_\setL$ is an inductive process on the ordinal $\alpha$. Initially (step $\alpha = 0$), each $v \in V$ is labeled by:
$$\lambda_0(v) = \{o_l \in \setP_\setL \mid \mbox{there exists a play } hl^{\omega} \mbox{ with } h \in \Hist(v) \mbox{ and } l \in \setL\}.$$
(In particular $\lambda_0(l) = \{o_l\}$ for all $l \in \setL$). By the first hypothesis of the theorem, $\lambda_0(v) \neq \emptyset$. Let us introduce some additional terminology. At step $\alpha$, when there is a path\footnote{By path, we mean a finite path} $\pi$ from $v$ to $v'$ in $G$, we say that $\pi$ is \emph{$(o,\alpha)$-labeled} if $o \in \lambda_{\alpha}(u)$ for all the vertices $u$ of $\pi$. Thus initially, we have a $(o_l,0)$-labeled path from $v$ to $l$ for each $o_l \in \lambda_0(v)$. For $v \in V$, let
$$m_{\alpha}(v) = \mbox{$\max_{\prec_{v}}$} \{ \mbox{$\min_{\prec_{v}}$} \lambda_{\alpha}(v') \mid v' \in \Succ(v)\}$$
with the convention that $m_\alpha(v) = \top$ if $\Succ(v) = \emptyset$ or if $\lambda_{\alpha}(v') = \emptyset$ for all $v' \in \Succ(v)$.\footnote{We suppose that $o \prec_v \top$ for all $o \in \setP_\setL$.} When $m_\alpha(v) \neq \top$, we says that $v' \in \Succ(v)$ \emph{realizes} $m_{\alpha}(v)$ if $m_{\alpha}(v) =\min_{\prec_{v}} \lambda_{\alpha}(v')$.  Notice that even if $\Succ(v)$ could be infinite, there are finitely many sets $\lambda_{\alpha}(v')$ since $\setP_\setL$ is finite. This justifies our use of $\max_{\prec_{v}}$ and $\min_{\prec_{v}}$ operators in the definition of $m_{\alpha}(v)$.

We alternate between applying \Remove\ and \Adjust\ to the current labeling. More formally, we define the labeling $\lambda_\alpha$ inductively\footnote{Note that our definition as written makes non-deterministic choices. This is immaterial for our purposes, but could be determinized by demanding a well-ordering of the vertex set and the outcomes.}. In the following, $\gamma$ is always assumed to be a limit ordinal and $n$ to be a natural number.

\begin{itemize}
\item {\bf Defining $\lambda_{\gamma + 2n + 1}$ via \Remove\ operation}

Let $\alpha := \gamma + 2n + 1$. Test if for some $v \in V$, there exist $o \in \lambda_{\alpha - 1}(v)$ and $v' \in \Succ(v)$ such that
$$o \prec_v o', \mbox{ for all } o' \in \lambda_{\alpha - 1}(v').$$

If such a $v$ exists, then $\lambda_{\alpha}(v) =  \lambda_{\alpha - 1}(v) \setminus \{o\}$, and $\lambda_{\alpha}(u) =  \lambda_{\alpha - 1}(u)$ for the other vertices $u \neq v$. Otherwise $\lambda_{\alpha}(u) =  \lambda_{\alpha - 1}(u)$ for all $u \in V$.

%\item {\bf Defining $\lambda_{\gamma + 2n + 1}$ via \Remove\ operation}

%Test if for some $v \in V$, there exist $o \in \lambda_{\gamma + 2n}(v)$ and $v' \in \Succ(v)$ such that
%$$o \prec_v o', \mbox{ for all } o' \in \lambda_{\gamma + 2n}(v').$$

%If such a $v$ exists, then $\lambda_{\gamma + 2n + 1}(v) =  \lambda_{\gamma + 2n}(v) \setminus \{o\}$, and $\lambda_{\gamma + 2n+1}(u) =  \lambda_{\gamma + 2n}(u)$ for the other vertices $u \neq v$. Otherwise $\lambda_{\gamma + 2n+1}(u) =  \lambda_{\gamma + 2n}(u)$ for all $u \in V$.

%\item {\bf Defining $\lambda_{\gamma + 2n + 2}$ via \Adjust\ operation}

%Suppose that $\lambda_{\gamma +2n+1}(v) =  \lambda_{\gamma + 2n}(v) \setminus \{o\}$ at the previous step. For all $u \in V$ such that $o \in \lambda_{\gamma +2n+1}(u)$, test if there exists a $(o,\gamma +2n+1)$-labeled path from $u$ to some $l \in \setL$. If yes, then $\lambda_{\gamma +2n+2}(u) =  \lambda_{\gamma +2n+1}(u)$, otherwise $\lambda_{\gamma +2n+2}(u) =  \lambda_{\gamma +2n+1}(u) \setminus \{o\}$. For all $u \in V$ such that $o \not\in \lambda_{\gamma +2n+1}(u)$, let $\lambda_{\gamma +2n+2}(u) =  \lambda_{\gamma +2n+1}(u)$.

%Suppose that $\lambda_{\gamma +2n+1}(v) =  \lambda_{\gamma +2n}(v)$ for all $v \in V$ at the previous step, then $\lambda_{\gamma +2n+2}(v) =  \lambda_{\gamma +2n+1}(v)$ for all $v \in V$.

\item {\bf Defining $\lambda_{\gamma + 2n + 2}$ via \Adjust\ operation}

Let $\alpha := \gamma + 2n + 2$. Suppose that $\lambda_{\alpha - 1}(v) =  \lambda_{\alpha - 2}(v) \setminus \{o\}$ at the previous step. For all $u \in V$ such that $o \in \lambda_{\alpha - 1}(u)$, test if there exists a $(o,\alpha - 1)$-labeled path from $u$ to some $l \in \setL$. If yes, then $\lambda_{\alpha}(u) =  \lambda_{\alpha - 1}(u)$, otherwise $\lambda_{\alpha}(u) =  \lambda_{\alpha - 1}(u) \setminus \{o\}$. For all $u \in V$ such that $o \not\in \lambda_{\alpha-1}(u)$, let $\lambda_{\alpha}(u) =  \lambda_{\alpha-1}(u)$.

Suppose that $\lambda_{\alpha-1}(v) =  \lambda_{\alpha-2}(v)$ for all $v \in V$ at the previous step, then $\lambda_{\alpha}(v) =  \lambda_{\alpha-1}(v)$ for all $v \in V$.

\item {\bf Defining $\lambda_\gamma$ via intersection}

Let $\lambda_{\gamma}(v) = \cap_{\beta < \gamma} \lambda_{\beta}(v)$ for all $v \in V$.

\end{itemize}

For each $v$, the sequence $(\lambda_\alpha(v))_{\alpha}$ is nonincreasing (w.r.t.~set inclusion), and thus the sequence $(m_\alpha(v))_{\alpha}$ is nondecreasing (w.r.t.~$\prec_v$). Moreover, the sequence $(\lambda_\alpha)_\alpha$ is nonincreasing w.r.t.~pointwise set inclusion. Thus, there exists some ordinal $\alpha^*$ such that $\lambda_{\alpha^*} = \lambda_\beta$ for all $\beta > \alpha^*$. By inspecting the definition, we see that it suffices to check that $\lambda_{\alpha^*} = \lambda_{\alpha^*+1} = \lambda_{\alpha^*+2}$ in order to see that $\alpha^*$ is a fixed point. If $V$ is finite, such a fixed point is reached after at most $2 |\setP_L| \cdot |V|$ steps. The central challenge is to show that this fixed point is non-empty in each component.

Notice that for all leaves $l \in \setL$ and all steps $\alpha$, we have $\lambda_\alpha(l) = \{o_l\}$.

\begin{lemma} \label{lem:fixpoint}
There exists an ordinal $\alpha^*$ such that
$$\lambda_{\alpha^*}(v) = \lambda_{\alpha^*+1}(v) = \lambda_{\alpha^*+2}(v)  \mbox{ for all } v \in V.$$
Moreover, $\lambda_{\alpha^*}(v) \neq \emptyset$ for all $v \in V$.
\end{lemma}

To be able to prove that $\lambda_{\alpha^*}(v) \neq \emptyset$, we introduce three invariants for which we will prove that they are initially true (Lemma~\ref{lemma:invinit}) and remain true after each step~$\alpha$ (Lemmata \ref{lemma:invremove},\ref{lemma:invadjust},\ref{lemma:invlimit}). The non emptiness of $\lambda_{\alpha^*}(v)$ will follow from the second invariant.

\begin{description}
\item[INV1] For $v \in V$, we have for all $v' \in \Succ(v)$ that
$$\{ o \in \lambda_{\alpha}(v') \mid m_{\alpha}(v) \preceq_{v} o \} \subseteq \lambda_{\alpha}(v).$$
In particular, when $m_{\alpha}(v) \neq \top$, for each $v'$ that realizes $m_{\alpha}(v)$, we have
\begin{eqnarray}\label{a(v)}
\lambda_{\alpha}(v') \subseteq \lambda_{\alpha}(v).
\end{eqnarray}
\item[INV2] For $v \in V$, $\lambda_{\alpha}(v) \neq \emptyset$.
\item[INV3] For $v \in V$, there exists a path from $v$ to some $l \in \setL$  such that for all vertices $u$ in this path, $\lambda_{\alpha}(u) \subseteq \lambda_{\alpha}(v)$.
%
%(Notice that the mentioned path is not supposed to be $(o_l,\alpha)$-labeled.)
\end{description}

\begin{lemma}
\label{lemma:invinit}
All three invariants are true for $\lambda_0$.
\end{lemma}
\begin{proof}
Consider $v \in V$ at the initial step $\alpha = 0$. By hypothesis there is a path from $v$ to some $l \in \setL$. Thus $\lambda_{\alpha}(v) \neq \emptyset$ and INV2 is true. Moreover, for all $v' \in \Succ(v)$, we have $\lambda_{\alpha}(v') \subseteq \lambda_{\alpha}(v)$ by the initial labeling, and thus INV1 and INV3 are also true.
\qed\end{proof}

\begin{lemma}
\label{lemma:invremove}
All three invariants are preserved by \Remove.
\end{lemma}
\begin{proof}
Consider some $\alpha = \gamma + 2n$ for limit ordinal $\gamma$ and $n \in \mathbb{N}$ such that all invariants hold for $\lambda_\alpha$. If $\lambda_{\alpha+1} = \lambda_{\alpha}$, then trivially, all invariants hold for $\lambda_{\alpha+1}$. Otherwise there exist $v$ and $o$ such that $\lambda_{\alpha+1}(v) = \lambda_{\alpha}(v) \setminus \{o\}$ and $\lambda_{\alpha+1}(u) = \lambda_{\alpha}(u)$ for all $u \neq v$. In particular $v \notin L$. For all $u \in V$, we have $m_{\alpha}(u) \preceq_{u} m_{\alpha +1}(u)$, with the particular case $m_{\alpha}(v) = m_{\alpha+1}(v)$.

\begin{itemize}
\item {\bf \Remove\ cannot violate INV1}. We first consider $u \in V$ such that $u \neq v$. For all $u' \in \Succ(u)$, we have
$$\begin{array}{llll}
\{o' \in \lambda_{\alpha+1}(u') \mid m_{\alpha+1}(u) \preceq_u o' \}\\
\subseteq \{o' \in \lambda_{\alpha}(u') \mid m_{\alpha}(u) \preceq_u o' \} & \mbox{since $\lambda_{\alpha+1}(u') \subseteq \lambda_{\alpha}(u')$} \\
& \mbox{and $m_{\alpha}(u) \preceq_{u} m_{\alpha +1}(u)$,}\\
\subseteq \lambda_{\alpha}(u) & \mbox{by INV1 at step $\alpha$,} \\
= \lambda_{\alpha+1}(u) & \mbox{as $u \neq v$.}
\end{array}$$

Let us turn to vertex $v$. As $o \prec_v m_{\alpha}(v)$, the previous inclusions can be modified as follows. For all $v' \in \Succ(v)$, we now have $\{o' \in \lambda_{\alpha+1}(v') \mid m_{\alpha+1}(v) \preceq_v o' \} \subseteq \{o' \in \lambda_{\alpha}(v') \mid m_{\alpha}(v) \preceq_v o' \} \subseteq \lambda_{\alpha}(v) \setminus \{o\} = \lambda_{\alpha+1}(v)$.

\item {\bf \Remove\ cannot violate INV2}. We only have to show that $\lambda_{\alpha + 1}(v) \neq \emptyset$. As $\Succ(v) \neq \emptyset$\footnote{Recall that $v \not\in \setL$, and that $\Succ(v) \neq \emptyset$ for all $v \in V \setminus \setL$.} and by INV2, we have $m_{\alpha}(v) \neq \top$. Hence there exists $v' \in \Succ(v)$ that realizes $m_{\alpha}(v) = m_{\alpha+1}(v)$. By INV1 and in particular $(\ref{a(v)})$ at step $\alpha +1$, we thus have $\lambda_{\alpha+1}(v') \subseteq \lambda_{\alpha+1}(v)$. As $\lambda_{\alpha+1}(v') = \lambda_{\alpha}(v') \neq \emptyset$, it follows that $\lambda_{\alpha+1}(v) \neq \emptyset$.

\item {\bf \Remove\ cannot violate INV3}.
We first consider $u \neq v$. By INV3, there exists a path $\pi$ from $u$ to some $l \in \setL$ such that $\lambda_{\alpha}(w) \subseteq \lambda_{\alpha}(u)$ for all vertices $w$ in this path. We can keep the path $\pi$ at step $\alpha+1$ since $\lambda_{\alpha+1}(w) \subseteq \lambda_{\alpha}(w)$ for all $w$ in $\pi$ and $\lambda_{\alpha+1}(u) = \lambda_{\alpha}(u)$.

We now consider vertex $v$. Consider again $v' \in \Succ(v)$ that realizes $m_{\alpha+1}(v)$. By $(\ref{a(v)})$, $\lambda_{\alpha+1}(v') \subseteq \lambda_{\alpha+1}(v)$. We know that there exists a path $\pi$ from $v'$ to some $l \in \setL$ such that $\lambda_{\alpha}(w) \subseteq \lambda_{\alpha}(v')$ for all $w$ in $\pi$. This path $\pi$ augmented with the edge $(v,v')$ is the required path for INV3 at step $\alpha+1$ because for all $w$ in $\pi$, we have $\lambda_{\alpha+1}(w) \subseteq \lambda_{\alpha}(w) \subseteq \lambda_{\alpha}(v') = \lambda_{\alpha+1}(v') \subseteq \lambda_{\alpha+1}(v)$.
\end{itemize}
\qed\end{proof}

\begin{lemma}
\label{lemma:invadjust}
All three invariants are preserved by \Adjust.
\end{lemma}
\begin{proof}
Let all three invariants hold for $\alpha = \gamma + 2n + 1$ for limit ordinal $\gamma$ and $n \in \mathbb{N}$. Then the preceding step was a \Remove\ step. If $\lambda_{\alpha} = \lambda_{\alpha-1}$, then $\lambda_{\alpha+1} = \lambda_\alpha$. Otherwise, there are $v_0 \in V$ and an outcome $o$ such that $\lambda_{\alpha}(v_0) = \lambda_{\alpha-1}(v_0) \setminus \{o\}$ and $\lambda_\alpha(u) = \lambda_{\alpha-1}(u)$ for all $u \neq v_0$.

For all $v \in V$, either $\lambda_{\alpha+1}(v) = \lambda_{\alpha}(v)$ or $\lambda_{\alpha+1}(v) = \lambda_{\alpha}(v) \setminus \{o\}$, and $m_{\alpha}(v) \preceq_v m_{\alpha+1}(v)$.

Consider $v \in V$ such that $o \notin \lambda_{\alpha+1}(v)$ and $o \in \lambda_{\alpha}(v)$. Then
\begin{eqnarray} \label{son}
\forall v' \in \Succ(v),~ o \not\in \lambda_{\alpha+1}(v')
\end{eqnarray}
Otherwise if $o \in \lambda_{\alpha+1}(v')$ for some $v' \in \Succ(v)$, this means that $o$ has not been removed from $\lambda_{\alpha}(v')$, i.e., there exists a $(o,\alpha)$-labeled path from $v'$ to some $l \in \setL$, and thus also from $v$ to $l$ by using the edge $(v,v')$. This is in contradiction with $o$ being removed from $\lambda_{\alpha}(v)$.

\begin{itemize}
\item {\bf \Adjust\ cannot violate INV1}. We first consider $v \in V$ such that $\lambda_{\alpha+1}(v) = \lambda_{\alpha}(v)$. As done for INV1 and \Remove, we have for all $v' \in \Succ(v)$ that $\{o' \in \lambda_{\alpha+1}(v') \mid m_{\alpha+1}(v) \preceq_v o' \} \subseteq \{o' \in \lambda_{\alpha}(v') \mid m_{\alpha}(v) \preceq_v o' \} \subseteq \lambda_{\alpha}(v) = \lambda_{\alpha+1}(v)$.

We now consider $v \in V$ such that $\lambda_{\alpha+1}(v) \neq \lambda_{\alpha}(v)$. Let $v' \in \Succ(v)$. From $(\ref{son})$, we have $\{o' \in \lambda_{\alpha+1}(v') \mid m_{\alpha+1}(v) \preceq_v o' \} \subseteq \{o' \in \lambda_{\alpha}(v') \mid m_{\alpha}(v) \preceq_v o' \} \setminus \{o \} \subseteq \lambda_{\alpha}(v) \setminus\{o\} = \lambda_{\alpha+1}(v)$.

\item {\bf \Adjust\ cannot violate INV2}. Assume that for some $v \in V$, $\lambda_{\alpha+1}(v) = \emptyset$, that is, $\lambda_{\alpha}(v) = \{o\}$. By INV3, there exists a path $\pi$ from $v$ to some $l \in \setL$ such that $\lambda_{\alpha}(u) \subseteq \lambda_{\alpha}(v)$ for all $u$ in $\pi$.  From $\lambda_{\alpha}(v) = \{o\}$ and $\lambda_{\alpha}(u) \neq \emptyset$ (by INV2), we get $\lambda_{\alpha}(u) = \{o\}$ for all such $u$. Therefore, the path $\pi$ from $v$ to $l$ is $(o,\alpha)$-labeled and $o$ cannot be removed from $\lambda_{\alpha}(v)$, showing that $\lambda_{\alpha+1}(v) \neq \emptyset$.

\item {\bf \Adjust\ cannot violate INV3}. Let $v \in V$ and by INV3 take a path $u_1 \ldots u_n$ from $v = u_1$ to some $l = u_n$ with $l \in \setL$ such that $\lambda_{\alpha}(u_i) \subseteq \lambda_{\alpha}(v)$ for all $i$. Either this path is still valid at step $\alpha+1$, or there exists a smallest $i$ such that $o \in \lambda_{\alpha+1}(u_i) = \lambda_{\alpha}(u_i)$, but $o \in \lambda_{\alpha}(v)$ and $o \not\in \lambda_{\alpha+1}(v)$. By minimality of $i$, $o \not\in \lambda_{\alpha+1}(u_j)$ for all $j \leq i-1$.

By the contraposition of (\ref{son}) with $u_{i-1}$ and $u_{i}$, knowing that $o \not\in \lambda_{\alpha+1}(u_{i-1})$, it follows that $o \not\in \lambda_{\alpha}(u_{i-1})$. By INV3 there is a path $\pi$ from $u_{i-1}$ to some $l' \in \setL$ such that for all $w$ in $\pi$, $\lambda_{\alpha}(w) \subseteq \lambda_{\alpha}(u_{i-1})$ $(\subseteq \lambda_{\alpha}(v))$. Notice that $o \not\in \lambda_{\alpha}(w)$ for all these $w$ since $o \not\in \lambda_{\alpha}(u_{i-1})$. The path $\pi'$ obtained by concatenating $u_1 \ldots u_{i-1}$ with $\pi$ is the required path from $v$ for INV3 at step $\alpha +1$. Indeed for all $w'$ in $\pi'$, we have seen that $\lambda_{\alpha}(w') \subseteq \lambda_{\alpha}(v)$ and $o \notin \lambda_{\alpha+1}(w')$. Thus $\lambda_{\alpha+1}(w') \subseteq \lambda_{\alpha}(v) \setminus \{o\}= \lambda_{\alpha+1}(v)$.
\end{itemize}
\qed\end{proof}

\begin{lemma}
\label{lemma:invlimit}
If all three invariants are true for each $\lambda_\beta$, $\beta < \alpha$, $\alpha$ a limit ordinal, then they are true for $\lambda_\alpha$.
\end{lemma}
\begin{proof}
Let $\alpha$ be a limit ordinal, and suppose that the three invariants are true for each ordinal $\beta < \alpha$. Given $v \in V$, as the set $\lambda_{\beta}(v)$ is finite\footnote{This is the place in the proof where finiteness of the number of outcomes is used in a crucial way.} and the sequence $(\lambda_{\beta}(v))_{\beta < \alpha}$ is nonincreasing, there exists some $\gamma < \alpha$ such that $\lambda_{\beta}(v) = \lambda_{\gamma}(v)$ for all $\beta$, $\gamma \leq \beta < \alpha$. Therefore
\begin{eqnarray} \label{eq:gamma}
\lambda_\alpha(v) = \cap_{\beta < \alpha} \lambda_{\beta}(v) = \lambda_{\gamma}(v).
\end{eqnarray}
It immediately follows that INV2 holds at step $\alpha$. To show that INV3 also holds, consider a path $\pi$ from $v$ to some $l \in \setL$ such that $\lambda_{\gamma}(u) \subseteq \lambda_{\gamma}(v)$ for all $u$ in $\pi$ (by INV3 at step $\gamma$). We can take this path $\pi$ for INV3 at step $\alpha$ since for all these $u$, we have $\lambda_{\alpha}(u) \subseteq \lambda_{\gamma}(u) \subseteq \lambda_{\gamma}(v) = \lambda_{\alpha}(v)$. Finally, the first invariant remains true at step $\alpha$ because for all $v' \in \Succ(v)$, we have
$$\begin{array}{llll}
\{ o \in \lambda_{\alpha}(v') \mid m_{\alpha}(v) \preceq_{v} o \} \\
\subseteq \{ o \in \lambda_{\gamma}(v') \mid m_{\gamma}(v) \preceq_{v} o \} & \mbox{since $\lambda_{\alpha}(v') \subseteq \lambda_{\gamma}(v')$ and $m_{\gamma}(v) \preceq_v m_{\alpha}(v)$,} \\
\subseteq \lambda_{\gamma}(v) & \mbox{by INV1 at step $\gamma$,} \\
=  \lambda_{\alpha}(v) & \mbox{by (\ref{eq:gamma}).}
\end{array}$$
\qed\end{proof}

To get Theorem~\ref{thm:generalgraph}, it remains to explain how to build a finite-memory weak SPE $\bar \sigma$ from the fixed point provided by Lemma~\ref{lem:fixpoint}.

\begin{proof}[of Theorem~\ref{thm:generalgraph}]
By Lemma~\ref{lem:fixpoint}, we have a fixed point of \Remove\ and \Adjust\ such that that $\lambda_{\alpha^*}(v) \neq \emptyset$ for all $v \in V$. Since $\lambda_{\alpha^*}$ is unchanged by \Adjust, for all $o \in \lambda_{\alpha^*}(v)$, there is a $(o,\alpha^*)$-labeled path $\pi$ from $v$ to some $l \in \setL$ with $o_l = o$. We denote by $\rho_{v,o}$ the play $\pi l^\omega$:
\begin{eqnarray} \label{eq:rhovl}
\rho_{v,o} = \pi l^\omega.
\end{eqnarray}
(*) Recall that $\mu(\rho_{v,o}) = o_l$, and have in mind that $o_l \in \lambda_{\alpha^*}(u)$ for all vertices $u$ in $\rho_{v,o}$.

The construction of $\bar \sigma$ will be done step by step thanks to a progressive labeling of the histories by outcomes in $\setP_\setL$ and by using the plays $\rho_{v,o}$. This labeling $\kappa : \Hist(v_0) \rightarrow \setP_\setL$ will allow to recover from history $hv$ the outcome $o$ of the play $\out{\Sub{\bar \sigma}{h}}_{v}$ induced by $\bar \sigma$ in the subgame $(\Sub{G}{h},v)$.

We start with history $v_0$ and any $o_0 \in \lambda_{\alpha^*}(v_0)$. Consider $\rho_{v_0,o_0}$ as in (\ref{eq:rhovl}). The strategy profile $\bar \sigma$ is partially built such that $\out{\bar \sigma}_{v_0} = \rho_{v_0,o_0}$. The non empty prefixes $g$ of $\rho_{v_0,o_0}$ are all labeled with $\kappa(g) = o_0$.

At the following steps, we consider a history $h'v'$ that is not yet labeled, but such that $h' = hv$ has already been labeled by $\kappa(hv) = o$.  The labeling of $hv$ by $o$ means that $\bar \sigma$ has already been built to produce the play $\out{\Sub{\bar \sigma}{h}}_{v}$ with outcome $o$ in the subgame $(\Sub{G}{h},v)$, such that $\out{\Sub{\bar \sigma}{h}}_{v}$ is suffix of $\rho_{u,o}$ from some $u$. By (*) we have $o \in \lambda_{\alpha^*}(v)$. As $\lambda_{\alpha^*}$ is invariant under \Remove\ (noting $o \in \lambda_{\alpha^*}(v)$ and $v' \in \Succ(v)$), there exists $o' \in \lambda_{\alpha^*}(v')$ such that
\begin{eqnarray} \label{eq:chosen}
o \nprec_v o'.
\end{eqnarray}
With $\rho_{v',o'}$ as in (\ref{eq:rhovl}), we then extend the construction of $\bar \sigma$ such that $\out{\Sub{\bar \sigma}{h'}}_{v'} = \rho_{v',o'}$, and for each non empty prefix $g$ of $\rho_{v',o'}$, we label $h'g$ by $\kappa(h'g) = o'$ (notice that the prefixes of $h'$ have already been labeled by choice of $h'$). This process is iterated to complete the construction of $\bar \sigma$.

Let us show that the constructed profile $\bar \sigma$ is a very weak SPE in $(G,v_0)$. Consider a history $h' = hv \in \Hist(v_0)$ with $v \in V_i$, and a one-shot deviating strategy $\sigma'_i$ from $\Sub{\sigma_i}{h}$ in the subgame $(\Sub{G}{h},v)$.
Let $v'$ be such that $\sigma'_i(v) = v'$. By definition of $\bar \sigma$, we have $\kappa(hv) = o$ and $\kappa(h'v') =  o'$ such that (\ref{eq:chosen}) holds. Let $\rho = \out{\Sub{\bar \sigma}{h}}_v$ and $\rho' = \out{\Sub{\bar \sigma}{h'}}_{v'}$.  Then $o = \mu(h\rho)$ and $o' = \mu(hv \rho')$ by (*). By (\ref{eq:chosen}), $\sigma'_i$ is not a profitable deviation for player $i$. Hence $\bar \sigma$ is a very weak SPE and thus a weak SPE by Proposition~\ref{prop:weak-veryweak}.

It remains to prove that $\bar \sigma$ is finite-memory by correctly choosing the plays $\rho_{v,o}$ of (\ref{eq:rhovl}). Fix $o \in \setP_\setL$ and consider the set $U_o$ of vertices $v$ such that $o \in \lambda_{\alpha^*}(v)$. Then we choose the plays $\rho_{v,o} = \pi l^\omega$ for all $v \in U_o$, such that the set of associated finite paths $\pi l$ forms a tree. Therefore having $o$ in memory, the required Moore machine can produce positionally each $\rho_{v,o}$ with $v \in U_o$. Hence its set $M$ of states is equal to $\setP_\setL$.
\qed\end{proof}

The next corollary is an easy consequence of Theorem~\ref{thm:generalgraph}. Under the same conditions except perhaps the second one, and when the underlying graph of $G$ is a tree, it guarantees the existence of a weak SPE that is positional.

\begin{corollary} \label{cor:generaltree}
Let $(G,v_0)$ be an initialized game with a subset $\setL \subseteq V$ of leaves\footnote{The existence of leaves $l$ with a unique outgoing edge $(l,l)$ is abusive since the graph is a tree: it should be understood as a unique infinite play from $l$.} such that the underlying graph is a tree rooted at $v_0$. If $(G,v_0)$ satisfies the first and third conditions of Theorem~\ref{thm:generalgraph}, then there exists a positional weak SPE in $(G, v_0)$.
\end{corollary}

\begin{proof}
If the second condition of Theorem~\ref{thm:generalgraph} is not satisfied, we replace the outcome function $\mu$ by a new function $\mu'$ defined as follows. For all plays $l^\omega$, with $l \in \setL$, there is a unique path $\pi$ from $v_0$ to $l$ as the underlying graph is a tree. For all suffixes $\rho$ of $\pi l^\omega$, we let $\mu'(\rho) = \mu(\pi l^\omega)$. For all the remaining plays $\rho$, we let $\mu'(\rho) = \mu(\rho)$. With the new function $\mu'$, the game $(G,v_0)$ now satisfies all the conditions of Theorem~\ref{thm:generalgraph} and has thus a weak SPE $\bar \sigma$ with respect to $\mu'$. It is easy to see that $\bar \sigma$ is also a weak SPE with respect to $\mu$. Notice that this profile is necessarily positional as the underlying graph is a tree.
\qed\end{proof}

In the next two sections, we present two large families of games for which there always exists a weak SPE. We will explain how these results are obtained from Theorem~\ref{thm:generalgraph} and its Corollary~\ref{cor:generaltree}. Before that, we demonstrate the argument establishing Theorem~\ref{thm:generalgraph} on the game $G_4$ as introduced in Example \ref{ex:Gn}.

\begin{example}
Let us describe the inductive process for the game $G_4$ of Figure~\ref{fig:G4} (Page \pageref{fig:G4}). For all $i \in \Pi$ and all steps $\alpha$, we have $\lambda_{\alpha}(l_i) = \{o_i\}$. Table~\ref{tab:steps} indicates the different steps until reaching $\alpha^*$ for the vertices $v_i$, $i \in \Pi$, with $\setP_\setL = \{o_1,o_2,o_3,o_4\}$. For instance, at step~$1$, \Remove\ removes $o_4$ from $\lambda_\alpha(v_1)$ because $o_4 \prec_1 o'$ for all $o' \in \lambda_\alpha(l_1) = \{o_1\}$. At step~2, \Adjust\ removes no outcome. For $v = v_1$ and $o \in \lambda_\alpha(v_1)$, the plays $\rho_{v,o}$ are:
$$\rho_{v_1,o_1} = v_1 l_1^{\omega}, \quad \rho_{v_1,o_2} = v_1 v_2 l_2^{\omega}, \quad \rho_{v_1,o_3} = v_1 v_2 v_3 l_3^{\omega}.$$
The other vertices $v \neq v_1$ have similar plays $\rho_{v,o}$.

\begin{table}
\begin{center}
$\begin{array}{|c|c|c|c|c|}
\hline
\alpha & \lambda_\alpha(v_1) & \lambda_\alpha(v_2) & \lambda_\alpha(v_3) & \lambda_\alpha(v_4) \\
\hline
0 & \setP_\setL & \setP_\setL & \setP_\setL & \setP_\setL \\
1 & \setP_\setL\setminus \{o_4\} & \setP_\setL & \setP_\setL & \setP_\setL \\
2 & \setP_\setL\setminus \{o_4\} & \setP_\setL & \setP_\setL & \setP_\setL \\
3 & \setP_\setL\setminus \{o_4\} & \setP_\setL \setminus \{o_1\} & \setP_\setL & \setP_\setL \\
4 & \setP_\setL\setminus \{o_4\} & \setP_\setL \setminus \{o_1\} & \setP_\setL & \setP_\setL \\
5 & \setP_\setL\setminus \{o_4\} & \setP_\setL \setminus \{o_1\} & \setP_\setL \setminus \{o_2\} & \setP_\setL \\
6 & \setP_\setL\setminus \{o_4\} & \setP_\setL \setminus \{o_1\} & \setP_\setL \setminus \{o_2\} & \setP_\setL \\
7 & \setP_\setL \setminus \{o_4\} & \setP_\setL \setminus \{o_1\} & \setP_\setL \setminus \{o_2\} & \setP_\setL \setminus \{o_3\} \\
\alpha^* = 8 & \setP_\setL \setminus \{o_4\} & \setP_\setL \setminus \{o_1\} & \setP_\setL \setminus \{o_2\} & \setP_\setL \setminus \{o_3\} \\
\hline
\end{array}$
\end{center}
\caption{The different steps until reaching a fixed point for game $G_4$}
\label{tab:steps}
\end{table}

In the case of game $(G_4,v_1)$, the construction of a weak SPE $\bar \sigma$, as described in the previous proof, leads to the strategy profile of Figure~\ref{fig:profileG4}. Indeed, the construction of $\bar \sigma$ begins with history $v_1$ and $\rho_{v_1,o_1} = v_1 l_1^\omega$. At the next step, we consider history $v_1v_2$ and $\rho_{v_2,o_4} = v_2v_3v_4l_4^\omega$ such that $o_1 \nprec_1 o_4$, aso. Notice that the previous proof states a memory size equal to~$4$ for $\bar \sigma$ whereas Figure~\ref{fig:MooreG4} depicts a Moore machine for $\bar \sigma$ with a better memory size equal to~$3$.
\end{example}

\section{First application} \label{sec:first}
%-------------------------------

In this section, we begin with the first application of the results of the previous section (more particularly Corollary~\ref{cor:generaltree}): when an initialized game has an outcome function with finite range, then it always has a weak SPE.

\begin{theorem}  \label{thm:infinitetree}
Let $(G,v_0)$ be an initialized game such that the outcome function has finite range. Then there exists a weak SPE in $(G,v_0)$.
\end{theorem}

Let us comment this theorem. \emph{(i)} Kuhn's theorem~\cite{kuhn53} states that there always exist an SPE in initialized games played on a \emph{finite tree} (notice that in this particular case, the existence of a weak SPE is equivalent to the existence of an SPE). Theorem~\ref{thm:infinitetree} can be seen as a  generalization of Kuhn's theorem: if we keep the outcome set finite, all initialized games (regardless of the underlying graph and the player set) have weak SPE. \emph{(ii)} The next theorem is proved in~\cite{Flesch10} for outcome functions $\mu = (\mu_i)_{i \in \Pi}$ as presented in Example~\ref{ex:classical} and has strong relationship with Theorem~\ref{thm:infinitetree}. Recall that a payoff function $\mu_i : \Plays \to \mathbb R$ is \emph{lower-semicontinuous} if whenever a sequence of plays $(\rho_n)_{n \in \mathbb N}$ converges to a play $\rho = \lim_{n \rightarrow \infty} \rho_n$, then $\liminf_{n \rightarrow \infty} \mu_i(\rho_n) \geq \mu_i(\rho)$.

%Since $V$ is endowed with the discrete topology, and thus $V^\omega$ with the product topology, a sequence of plays $(\rho_n)_{n \in \IN}$ converges to a play $\rho = \lim_{n \rightarrow \infty} \rho_n$ if every prefix of $\rho$ is prefix of all $\rho_n$ except, possibly, of finitely many of them.

\begin{theorem}[\cite{Flesch10}]
Let $(G,v_0)$ be an initialized game with a finite set $\Pi$ of players and an outcome function $\mu = (\mu_i)_{i \in \Pi}$ such that each $\mu_i : \Plays \to \mathbb R$ has finite range and is lower-semicontinuous. Then there exists an SPE in $(G,v_0)$.
\end{theorem}

\noindent As every weak SPE is an SPE in the case of lower-semicontinuous payoff functions $\mu_i$~\cite{BBMR15}, we recover the previous result with our Theorem~\ref{thm:infinitetree}. Even if it is not explicitly mentioned in~\cite{Flesch10}, a close look at the details of the proof shows that the authors first show the existence of a weak SPE (without the hypothesis of lower-semicontinuity) and then show that it is indeed an SPE (thanks to this hypothesis). The first part of their proof could be replaced by ours, which is simpler: we remove outcomes from the sets $\lambda_{\alpha}(v)$ (see the proof of Theorem~\ref{thm:generalgraph}) whereas plays are removed in the inductive process of \cite{Flesch10}. %Indeed our proof of Theorem~\ref{thm:generalgraph} deals with sets $\lambda_{\alpha}(v)$ of potential payoffs of outcome of weak SPE at step $\alpha$, whereas in~\cite{Flesch10}, the authors deals with potential outcomes of weak SPE which leads to a more complex proof.

\subsection{Intermediate results}
%-----------------------------------------
The proofs of Theorem~\ref{thm:infinitetree} in this section and Theorem~\ref{thm:finitegraph} in the next section require several intermediate results that we now describe. We begin with the next lemma where the set $\mu^{-1}(\{o\})$, with $o \in \setP$, is said to be \emph{dense in $(G,v_0)$} if for all $h \in \Hist(v_0)$, there exists $\rho$ such that $h\rho$ is a play with outcome $\mu(h\rho) = o$.

\begin{lemma} \label{lem:dense}
Let $(G,v_0)$ be an initialized game. If for some $o \in \setP$, the set $\mu^{-1}(\{o\})$ is dense in $(G,v_0)$, then there exists a weak SPE with outcome $o$ in $(G,v_0)$.
\end{lemma}

\begin{proof}
The construction of a very\footnote{As already done before, we apply Proposition~\ref{prop:weak-veryweak}. It will be the case in the sequel of the article without mentioning anymore this proposition.} weak SPE $\bar \sigma$ is done step by step thanks to a progressive marking of the histories $hv \in \Hist(v_0)$. Let us give the construction of $\bar \sigma$. Initially, for history $v_0$, we know by density that there exists $\rho_0 \in \Plays(v_0)$ with outcome $o$. We partially construct $\bar \sigma$ such that it produces $\rho_0$, and we mark each non empty prefix of $\rho_0$. Then we consider a shortest unmarked history $hv$, and we choose some $\rho \in \Plays(v)$ such that $\mu(h\rho) = o$ (this is possible by density). We continue the construction of $\bar \sigma$ such that it produces the play $\rho$ in $(\Sub{G}{h},v)$, and for each non empty prefix $g$ of $\rho$, we mark $hg$ (notice that the prefixes of $h$ have already been marked by choice of $h$), and so on. In this way, we get a strategy profile $\bar \sigma$ in $(G,v_0)$ that is a weak SPE because in each subgame $(\Sub{G}{h},v)$, the play $\rho$ induced by $\Sub{\bar \sigma}{h}$ has outcome $\mu(h\rho) = o$ and each one-shot deviating strategy in $(\Sub{G}{h},v)$ leads to a play with outcome $o$.
\qed\end{proof}

Lemma~\ref{lem:dense} leads to the next two corollaries. The first one states the existence of a uniform weak SPE in each initialized game $(G,v)$, $v \in V$, when the underlying graph of $G$ is strongly connected and the outcome function is prefix-independent. This corollary will provide a first step towards Theorem~\ref{thm:finitegraph} presented in Section~\ref{sec:second}; it is already interesting on its own right.

\begin{corollary} \label{cor:SCC}
Let $G$ be a game such that the underlying graph is strongly connected and the outcome function $\mu$ is prefix-independent.
\begin{itemize}
\item Then for all realizable outcomes $o$ such that $o = \mu(\rho)$ with $\rho \in \Plays(v_0)$, there exists a weak SPE with outcome $o$ in $(G,v_0)$.
\item Moreover, there exists a uniform strategy profile $\bar \sigma$ and an outcome $o$ such that for all $v \in V$ taken as initial vertex, $\bar \sigma$ is a weak SPE in $(G,v)$ with outcome~$o$.
\end{itemize}
\end{corollary}
%
%Notice that Lemma~\ref{lem:dense} guarantees the existence of a weak SPE in each $(G,v)$ with $v \in V$. Indeed take $v_0 \in V$ and any simple cycle $\pi_0 v_0$ from $v_0$ to $v_0$. Such a cycle exists since the underlying graph is strongly connected. Let $\rho = \pi_0^{\omega}$ and $p = \mu(\rho)$ its outcome. For all $hu \in \Hist(v)$, there exists a path $\pi v_0$ from $u$ to $v_0$ as the underlying graph is strongly connected. The play $h\pi\rho$ has outcome equal to $\mu(\rho) = p$ since $\mu$ is prefix-independent. Hence $\mu^{-1}(p)$ is dense in $(G,v)$ and there exists a weak SPE in $(G,v)$ by Lemma~\ref{lem:dense}. Nevertheless to prove Corollary~\ref{cor:SCC}, we need to go further by exhibiting a uniform weak SPE with outcome~$p$ independently of the initial vertex $v$.

\begin{proof}
For the first statement, take $\rho \in \Plays(v_0)$ such that $o = \mu(\rho)$. By Lemma~\ref{lem:dense}, it is enough to show that $\mu^{-1}(\{o\})$ is dense in $(G,v_0)$ to get a weak SPE in $(G,v_0)$. For all $hv \in \Hist(v_0)$, there exists a path $\pi v_0$ from $v$ to $v_0$ as the underlying graph is strongly connected. The play $h\pi\rho$ has outcome equal to $\mu(\rho) = o$ since $\mu$ is prefix-independent. Hence $\mu^{-1}(\{o\})$ is dense.

To get the second statement, we need to go further by exhibiting a uniform weak SPE with the same outcome $o$ independently of the initial vertex $v$. Take any simple cycle $\pi_0v_0$ from $v_0$ to $v_0$. Such a cycle exists since the underlying graph is strongly connected. Let $\rho = \pi_0^\omega$ and $o = \mu(\rho)$ be its outcome. We partially construct a positional strategy profile $\bar \sigma$ that produces $\pi_0^{\omega}$ (recall that $\pi_0$ is simple). Let $U$ be the set of vertices that belong to $\pi_0$. Then extend the construction of $\bar \sigma$ to all $v \in V \setminus U$ in a way to reach $U$ (i.e. the cycle $\pi_0$) positionally. We then get the required uniform strategy profile $\bar \sigma$ with outcome~$o$.
\qed\end{proof}

The second corollary is a generalization of the previous one. It still guarantees the existence of a uniform weak SPE in all games $(G,v)$, $v \in V$, for graphs that are not necessarily strongly connected but have bottom strongly connected components all containing a play induced by a simple cycle and with the same outcome. This result will be useful in the proof of Theorem~\ref{thm:uniform} in Section~\ref{sec:second}.

\begin{corollary} \label{cor:SCCwithsamepayoff}
Let $G$ be a game such that the underlying graph is finite and the outcome function $\mu$ is prefix-independent. Suppose that there exists an outcome $o$ such that in each bottom strongly connected component $C$ of $G$, one can find a play $\rho_C \in \Plays(v)$ for some $v \in C$ such that $\mu(\rho_C) = o$ and $\rho_C$ is induced by a simple cycle. Then there exists a uniform weak SPE with outcome $o$ in $(G,v)$, for all $v \in V$.
\end{corollary}

\begin{proof}
Let $\BotSCC$ be the set of bottom strongly connected components of $G$. The construction of the strategy profile $\bar \sigma$ is very close to the one proposed in the previous proof. We partially construct $\bar \sigma$ in a way to produce each $\rho_C$. This is possible positionally since each $\rho_C$ is induced by a simple cycle. Let $U$ be the set of vertices that belong to $\cup_{C \in \BotSCC}\rho_C$. Then extend the construction of $\bar \sigma$ to all $v \in V \setminus U$ in a way to reach $U$ positionally. This is possible by definition of $\cal C$. The resulting strategy profile $\bar \sigma$ is uniform and is a weak SPE in each $(G,v)$, $v \in V$, such that $\mu(\out{\bar \sigma}_{v}) = o$. Indeed each $\rho_C$ has outcome $o$ and  $\mu$ is prefix-independent.
\qed\end{proof}

We end with a last lemma which indicates how to combine different weak SPEs into one weak SPE. It will be used in the proofs of Theorems~\ref{thm:infinitetree} and~\ref{thm:finitegraph}. %\textcolor{blue}{The next explanation needs to be rephrased. Suggestion ?} The combined weak SPEs come from weak SPEs of some subgames $(\Sub{G}{h},l)$ of the given game $(G,v_0)$ and from a weak SPE of another game $(G',v_0)$. This game $(G',v_0)$ is constructed from $(G,v_0)$ by summarising each subgame $(\Sub{G}{h},l)$ and the payoff $p_{hl}$ of its weak SPE by a loop $(l,l)$ and a play $hl^{\omega}$ with payoff $p_{hl}$.

\begin{lemma} \label{lem:cutting}
Consider an initialized game $(G,v_0)$ and a set of vertices $\setL \subseteq V$ such that for all $hl \in \Hist(v_0)$ with $l \in \setL$, the subgame $(\Sub{G}{h},l)$ has a weak SPE with outcome $o_{hl}$. Consider another initialized game $(G',v_0)$ obtained from $(G,v_0)$
\begin{itemize}
\item by replacing all edges $(l,v) \in E$ by one edge $(l,l)$, for all $l \in \setL$,
\item and with outcome function $\mu'$ such that for all $\rho' \in \Plays_{G'}(v_0)$, $\mu'(\rho') = o_{hl}$ if $\rho' = hl^\omega$ with $l \in L$ and $\mu'(\rho') = \mu(\rho')$ otherwise.
\end{itemize}
If $(G',v_0)$ has a weak SPE, then $(G,v_0)$ has also a weak SPE.
\end{lemma}

\begin{proof}
Denote by $\bar \sigma^{hl}$ the weak SPE in each $(\Sub{G}{h},l)$, and by $\bar \sigma'$ the weak SPE in $(G',v_0)$. We then build a strategy profile $\bar \tau$ in $(G,v_0)$ as follows. For player~$i \in \Pi$ and history $hv \in \Hist_i(v_0)$:
\begin{itemize}
\item if no vertex of $\setL$ occurs in $hv$, then $\tau_i(hv) = \sigma'_i(hv)$;
\item otherwise, decompose $hv$ as $h_1h_2v$ such that the first occurrence of a vertex $l \in \setL$ is the first vertex of $h_2$. Then $\tau_i(hv) = \sigma^{h_1l}_i(h_2v)$.
\end{itemize}
Hence in the first case, $\tau_i$ mimics $\sigma'_i$ in the game $(G',v_0)$, and in the second case, $\tau_i$ mimics $\sigma^{h_1l}$ in the subgame $(\Sub{G}{h_1},l)$.

Let us show that $\bar \tau$ is a weak SPE in $(G,v_0)$. Consider any subgame $(\Sub{G}{h},v)$ such that $v \in V_i$, and any one-shot deviation strategy $\tau'_i$ of player~$i$ from $\Sub{\bar \tau}{h}$. Either no vertex of $\setL$ occurs in $hv$, and $\tau'_i$ is not profitable for player~$i$ because $\bar \sigma'$ is a weak SPE in $(G',v_0)$ and by definition of $\mu'$. Or $h = h_1h_2v$ such that the first occurrence of a vertex $l \in \setL$ is the first vertex of $h_2$, and again $\tau'_i$ is not profitable because $\bar \sigma^{h_1l}$ is a weak SPE in the subgame $(\Sub{G}{h_1},l)$.
\qed\end{proof}

\subsection{Proof of Theorem~\ref{thm:infinitetree}}
%-----------------------------------------------------------------

Now that we have established all useful intermediate results for this section and the next one, we can finally proceed to the proof of Theorem~\ref{thm:infinitetree}. W.l.o.g. we can suppose that the underlying graph of $G$ is a tree rooted at $v_0$ (by unraveling this graph from $v_0$). We first show how to transform a game played on an infinite tree to a game satisfying Conditions 1 and 2 from Theorem~\ref{thm:generalgraph} while reflecting weak SPE.

In the following lemma we write $h \sqsubseteq l$ to denote that $h$ is a prefix of $l$, and denote by $\mathrm{cl} (A)$ the topological closure of $A$.

\begin{lemma}
\label{lemma:bairecategory}
Consider a game played on an infinite tree $C^\omega$ with countable outcome set $O$ and outcome function $\mu : C^\omega \to O$. There exists a prefix-free set $L \subseteq C^*$ of \emph{leaves} and an assignment $\Theta : L \to O$ such that
\begin{enumerate}
\item For each $h \in C^*$ there exists some $l \in L$ with $h \sqsubseteq l$ or $l \sqsubseteq h$.
\item For each $l \in L$ we find that $\mu^{-1}(\{\Theta(l)\})$ is dense in $lC^\omega$.
\end{enumerate}
\end{lemma}
\begin{proof}
By iterative use of the Baire Category Theorem. We go through all $h \in C^*$ in some order, add elements to $L$ and extend $\Theta$. Let $h \in C^*$ be the current candidate. If we do not yet have added $l$ to $L$ with $l \sqsubseteq h$ or $h \sqsubseteq l$, then consider that $hC^\omega = hC^\omega \cap \bigcup_{o \in O} \mu^{-1}(\{o\})$. As $O$ is countable, the Baire Category Theorem implies that some $\mu^{-1}(\{o_0\})$ is somewhere dense, i.e.~that there exists some $l \sqsupseteq h$ such that $lC^\omega \subseteq \mathrm{cl} \left ( \mu^{-1}(\{o_0\}) \right )$. We add $l$ to $L$ and set $\Theta(l) = o_0$. Then we proceed to the next $h$. In the limit, we have constructed $L$ and $\Theta$ as desired.

To see that $L$ is prefix-free, assume that there are $l_1, l_2 \in L$ with $l_1 \sqsubset l_2$. If $l_1$ was added first, and $l_2$ was added when dealing with the history $h$, then $h \sqsubseteq l_2$. But as prefixes of a given history are linearly ordered, either $h \sqsubseteq l_1$ or $l_1 \sqsubseteq h$ follows. Thus, we would not have added $l_2$ when dealing with $h$. If $l_2$ was added first, and then $l_1$ when dealing with $h$, then we would find that $h \sqsubseteq l_1 \sqsubseteq l_2$, thus $h \sqsubseteq l_2$, thus we would not have added $l_2$. Hence, $L$ is prefix-free.
\qed\end{proof}

\begin{proof}[of Theorem~\ref{thm:infinitetree}]
Instead of reasoning with the underlying graph of $G$, we work w.l.o.g. %\footnote{By definition, having a weak SPE in $(G,v_0)$ is equivalent to having a (positional) weak SPE in its unraveling.}
with its unraveling from the initial vertex $v_0$.

We can apply Lemma \ref{lemma:bairecategory} to transform the game. For each leaf, we can apply Lemma~\ref{lem:dense} to obtain a weak SPE in the corresponding subgame. Together, the criteria of Lemma~\ref{lem:cutting} are satisfied. The implication of Lemma~\ref{lem:cutting} is true by Corollary~\ref{cor:generaltree}, and the conclusion yields the desired statement.
\qed\end{proof}

\section{Second application} \label{sec:second}
%----------------------------------

In this section, we present a second large family of games with a weak SPE, as another application of the general results of Section~\ref{sec:general} (more particularly Theorem~\ref{thm:generalgraph}). This family is constituted with all games with a finite underlying graph and a prefix-independent outcome function.

\begin{theorem} \label{thm:finitegraph}
Let $(G,v_0)$ be an initialized game such that the underlying graph is finite and the outcome function is prefix-independent. Then there exists a weak SPE in $(G,v_0)$.
%Then there exists a finite-memory weak SPE in $(G,v_0)$ with memory size bounded by the number of bottom strongly connected components of the graph. Moreover, a memory size linear in the number of bottom components is necessary.
\end{theorem}

Let us comment this theorem. \emph{(i)} It guarantees the existence of a weak SPE for classical games with \emph{quantitative} objectives as presented in Example~\ref{ex:classical}, such that their outcome function is prefix-independent. This is the case of \emph{limsup} and \emph{mean-payoff} payoff functions (and their limit inferior counterparts). Recall that Example~\ref{ex:contrex} (see also Figure~\ref{fig:gameNoSPE}) provides a game with no SPE, where the payoff functions $\mu_i$ can be seen as either \emph{limsup} or \emph{mean-payoff} (or their limit inferior counterparts). \emph{(ii)} Later in this section, we will show that under the hypotheses of Theorem~\ref{thm:finitegraph}, there always exists a weak SPE that is \emph{finite-memory} (Corollary~\ref{cor:finite-mem}), and we will study in which cases it can be \emph{positional} or even \emph{uniform}  (Theorem~\ref{thm:uniform}).  \emph{(iii)} The families of games of Theorems~\ref{thm:infinitetree} and~\ref{thm:finitegraph} are incomparable: Boolean reachability games are in the first family but not in the second one, and mean-payoff games are in the second family but not in the first one.

\subsection{Proof of Theorem~\ref{thm:finitegraph}}
%------------------------------------------------------------------

The proof of Theorem~\ref{thm:finitegraph} follows the same structure as for Theorem~\ref{thm:infinitetree}. The idea is to apply Lemma~\ref{lem:cutting} where $\setL$ is equal to the union of the bottom strongly connected components of the graph of $G$. The weak SPEs required by Lemma~\ref{lem:cutting} exist on the subgames $(\Sub{G}{h},l)$ with $l \in \setL$ by Corollary~\ref{cor:SCC}, and on the game $(G',v_0)$ thanks to Theorem~\ref{thm:generalgraph}.

\begin{proof}[of Theorem~\ref{thm:finitegraph}]
Let $\BotSCC$ be the set of bottom strongly connected components of the finite graph of $G$. By Corollary~\ref{cor:SCC}, for all $C \in \BotSCC$, there exist a uniform strategy profile $\bar \sigma_C$ and a outcome $o_C$ such that $\bar \sigma_C$ is a weak SPE with outcome $o_C$ in each $(G,v)$ with $v \in C$. Notice that as $\mu$ is prefix-independent, $\bar \sigma_C$ is also a weak SPE with outcome $o_C$ in all subgames $(\Sub{G}{h},v)$ with $hv \in \Hist(v_0)$ and $v \in C$.

If the initial vertex $v_0$ belongs to some $C \in \BotSCC$, then $\bar \sigma_C$ is the required weak SPE in $(G,v_0)$ (it is clearly finite-memory as it is uniform). From now on we suppose that $v_0 \not\in C$ for all $C \in \BotSCC$.

We consider the graph $(G',v_0)$ constructed from $(G,v_0)$ as described in Lemma~\ref{lem:cutting} with $\setL = \cup_{C \in \BotSCC}C$. This graph satisfies all the hypotheses of Theorem~\ref{thm:generalgraph}. 

The set $\setL$ of leaves is the one used for Lemma~\ref{lem:cutting}. The first hypothesis holds because $\setL$ is the union of the bottom strongly connected components of $G$. The second hypothesis holds because $\mu$ is prefix-independent. The third hypothesis holds because $V$ is finite. Therefore, $(G',v_0)$ has a weak SPE $\bar \sigma'$ by Theorem~\ref{thm:generalgraph}.

By the existence of the previous strategy profiles $\bar \sigma'$ and $\bar \sigma_C$, $C \in \BotSCC$, it follows by Lemma~\ref{lem:cutting} that there exists a weak SPE $\bar \tau$ in $(G,v_0)$.
\qed\end{proof}

\subsection{Finite-memory weak SPE} \label{subsec:finite-mem}
%------------------------------------------------

We here make the statement of Theorem~\ref{thm:finitegraph} more precise by guaranteeing the existence of a weak SPE with finite-memory.

\begin{corollary} \label{cor:finite-mem}
Let $(G,v_0)$ be an initialized game such that the underlying graph is finite and the outcome function is prefix-independent.
Then there exists a finite-memory weak SPE in $(G,v_0)$ with memory size bounded by the number of bottom strongly connected components of the graph. Moreover, a memory size linear in the number of bottom components is necessary.
\end{corollary}

\begin{proof}
In the proof of Theorem~\ref{thm:finitegraph}, we have constructed a weak SPE $\bar \tau$. Let us show that $\bar \tau$ is a finite-memory strategy profile with memory size bounded by $|\BotSCC|$. Let us first come back to the construction of $\bar \tau$ given in the proof of Lemma~\ref{lem:cutting}. Consider player~$i \in \Pi$ and history $hv \in \Hist_i(v_0)$. If no vertex of $\setL$ occurs in $hv$, then $\tau_i(hv) = \sigma'_i(hv)$. Otherwise, decompose $hv$ as $h_1h_2v$ such that the first occurrence of a vertex $l \in C \subseteq \setL$ is the first vertex of $h_2$, then
\begin{eqnarray} \label{eq:l-C}
\tau_i(hv) = \sigma_{C,i}(v).
\end{eqnarray}
Notice that in (\ref{eq:l-C}) $\tau_i(hv)$ only depends on $C$, and not on $l \in C$, since $\bar \sigma_C$ is uniform.
Now let us recall the construction of $\bar \sigma'$ with a memory size $|\setL|$ given in the proof of Theorem~\ref{thm:generalgraph}, and in particular to equation (\ref{eq:rhovl}). In $(G',v_0)$ the plays $\rho_{v,o} = \pi l^\omega$ can be produced positionally while keeping $l \in \setL$ in memory. Therefore by $(\ref{eq:l-C})$ and as $\bar \sigma_C$ is uniform, it follows that the memory size of $\bar \tau$ can be reduced from $|\setL|$ to $|\BotSCC|$.

Let us now prove that there exist games with a finite set $V$ and a prefix-independent function $\mu$, that require a memory size in $O(|\BotSCC|)$ for their weak SPEs. To this end, we come back to the family of games $G_n$ of Example~\ref{ex:Gn} with $n$ bottom strongly connected components. Consider the unravelling of $G_n$ from the initial vertex $v_1$ as depicted in Figure~\ref{fig:profileG4} and let us study the form of any weak SPE $\bar \sigma$ in $(G_n,v_1)$. In all subgames $(\Sub{G_n}{h},v_i)$, the induced play cannot be $(v_iv_{i+1} \ldots v_{i-1})^\omega$ with outcome $\bot$ since each player would have a profitable one-shot deviation. W.l.o.g let us suppose that $\sigma_1(v_1) =  l_1$ (player~$1$ decides to move from $v_1$ to $l_1$ at the root of the unravelling, as in Figure~\ref{fig:profileG4}). Then the outcome of the play $\rho$ induced by $\Sub{\bar\sigma}{v_1}$ in the subgame $(\Sub{G_n}{v_1},v_2)$ is necessarily $o_1$ or $o_n$, otherwise player~$1$ would have a profitable one-shot deviation in $(G_n,v_0)$ (recall that $o_1 \prec_1 o_j$ for all $j \in \Pi \setminus \{1,n\}$). The first case $o_1$ cannot occur otherwise player~$2$ would have a profitable one-shot deviation in $(\Sub{G_n}{v_1},v_2)$ (recall that $o_1 \prec_2 o_2$). With similar arguments one can verify that the induced play $\rho$ is necessarily equal to $v_2v_3 \ldots v_nl_n^\omega$ with outcome $o_n$ (as in Figure~\ref{fig:profileG4}). We can repeat the same reasoning for the play induced by $\Sub{\bar\sigma}{v_1v_2 \cdots v_n}$ in the subgame $(\Sub{G_n}{v_1v_2 \cdots v_n},v_1)$ which must be equal to $v_1v_2 \ldots v_{n-1}l_{n-1}^\omega$ with outcome $o_{n-1}$, aso. Hence all weak SPEs of $(G_n,v_1)$ have the form of the one described in Figure~\ref{fig:profileG4} and they have finite memory of size $n-1$ as explained previously in Example~\ref{ex:Gn} (see also Figure~\ref{fig:MooreG4}). Let us show that such a weak SPE $\bar \sigma$ cannot have a memory size $< n-1$. Assume the contrary: wlog consider the previous weak SPE $\bar \sigma$ (as in Figure~\ref{fig:profileG4}) and in particular a Moore machine ${\cal M} = (M,m_0,\alpha_U,\alpha_N)$ encoding $\sigma_1$ such that $|M| < n-1$. Let $h_jv_1$, $j \in \{0, \ldots, n-1\}$ be consecutive histories, with $h_j = (v_1v_2 \cdots v_n)^j$. On one hand, we have $\sigma_1(h_jv_1) = \alpha_N(\widehat{\alpha}_U(m_0,h_j),v_1)$ for all $j$. On the other hand, $\sigma_1(h_0v_1) = \sigma_1(h_{n-1}v_1) = l_1$ and $\sigma_1(h_jv_1) = v_2$ for all $j \in \{1, \ldots, n-2\}$. Therefore there exists $j_1, j_2 \in \{1, \ldots, n-2\}, j_1 \neq j_2$, such that the associated memory state is identical, i.e, $\widehat{\alpha}_U(m_0,h_{j_1}) = \widehat{\alpha}_U(m_0,h_{j_2})$. Thus $\cal M$ enters into a cycle while reading the prefixes of $(v_1v_2 \cdots v_n)^\omega$. This means that $\cal M$ defines $\sigma_1(hv) = v_2$ for all histories $h$ of which $h_1$ is prefix, in contradiction with $\sigma_1(h_{n-1}v_1) = l_1$.
\qed\end{proof}

\subsection{Positional weak SPE} \label{subsec:positional}
%------------------------------------------

In the previous section, Corollary~\ref{cor:finite-mem} guarantees the existence of a finite-memory weak SPE for games with a finite underlying graph and a prefix-independent outcome function. In this section, we identify conditions on the preference relations of the players, as expressed in the next lemma, that guarantee the existence of a \emph{uniform} weak SPE (see Theorem~\ref{thm:uniform}). %These conditions are necessary as explained in Example~\ref{ex:opt} below.

\begin{lemma}[Lemma 4 of \cite{Roux15}] \label{lem:killer}
Let $\setP$ be a non empty set of outcomes. Let $\prec_i$ be a preference relation over $\setP$, for all $i \in \Pi$. The following assertions are equivalent.
\begin{itemize}
\item For all $i, i' \in \Pi$ and all $o, p, q \in \setP$, we have $\neg(o \prec_i p \prec_i q \wedge q \prec_{i'} o \prec_{i'} p)$.
\item There exist a partition $\{\setP_k \}_{k \in \setK}$ of $\setP$ and a linear order $<$ over $\setK$ such that
\begin{itemize}
\item $k < k'$ implies $o \prec_i o'$ for all $i \in \Pi$, $o \in \setP_k$ and $o' \in \setP_{k'}$,
\item $\Sub{\prec_i}{\setP_k} = \Sub{\prec_{i'}}{\setP_k}$ or $\Sub{\prec_i}{\setP_k} = (\Sub{\prec_{i'}}{\setP_k})^{-1}$ for all $i, i' \in \Pi$.
\end{itemize}
\end{itemize}
\end{lemma}

In the previous lemma, we call each set $\setP_k$ a \emph{layer}. The second assertion states that $(i)$ if $k < k'$ then all outcomes in $\setP_{k'}$ are preferred to all outcomes in $\setP_k$ by all players, and $(ii)$ inside a layer, any two players have either the same preference relations or the inverse preference relations.
When a set of outcomes satisfies the conditions of Lemma~\ref{lem:killer}, we say that it is \emph{layered}. In~\cite{Roux15}, the author characterizes the preference relations that always yield SPE in games with outcome functions in the Hausdorff difference hierarchy of the open sets. One condition is that the set of outcomes is layered.

\begin{theorem} \label{thm:uniform}
Let $G$ be a game with a finite underlying graph and such that the outcome function is prefix-independent with a layered set $\setP$ of outcomes.
Then there exists a uniform weak SPE in $(G,v)$, for all $v \in V$.
\end{theorem}

\begin{example} \label{ex:opt}
Remember the class $G_n$ of games, $n \geq 3$, of Example~\ref{ex:Gn}, such that $\setP = \{o_1, \ldots, o_n, \bot \}$ and each player $i$ has a preference relation $\prec_i$ satisfying $\bot \prec_i o_{i-1} \prec_i o_i \prec_i o_j$ for all $j \in \Pi \setminus \{i-1,i\}$. This set of outcomes is not layered because the first assertion of Lemma~\ref{lem:killer} is not satisfied. Indeed we have
$$o_2 \prec_3 o_3 \prec_3 o_1 \wedge o_1 \prec_2 o_2 \prec_2 o_3.$$
Recall that in the proof of Corollary~\ref{cor:finite-mem} we have shown that all weak SPEs of $G_n$ require a memory size in $O(n)$. Hence the hypothesis of Theorem~\ref{thm:uniform} about the preference relations is not completely dispensable.
\end{example}

Let us proceed to the proof of Theorem~\ref{thm:uniform}.
Let $\BotSCC$ be the set of the bottom strongly connected components of the finite underlying graph of $G$. For each $C \in \BotSCC$, we fix a play $\rho_C \in \Plays(v)$ for some $v \in C$ induced by a simple cycle. The set $\setP_{\cal C} = \{o_C  \mid o_C = \mu(\rho_C), C \in \BotSCC \}$ is finite. It is layered by hypothesis with a finite partition into layers $\{\setP_k \}_{k \in \setK}$. The proof of Theorem~\ref{thm:uniform} is by induction on the number of layers and uses the next lemma dealing with one layer.

\begin{lemma} \label{lem:onelayer}
Suppose that $|K| = 1$, then there exists a uniform strategy profile~$\bar \sigma$ that is a weak SPE in each $(G,v)$, $v \in V$, such that $\mu(\out{\bar \sigma}_v) = o_C$ for some $C \in \BotSCC$.
\end{lemma}

The proof of this lemma is by induction on $|\setP_{\cal C}|$. The case of only one outcome is solved by Corollary~\ref{cor:SCCwithsamepayoff}. When they are several outcomes in $\setP_{\cal C}$, we will show how to decompose $G$ into two subgames $G'$ and $G''$ such that the bottom strongly connected component of $G'$ (resp. $G''$) are those components $C \in \cal C$ of $G$ such that $o_C = o$ for some $o$ (resp. $o_C \in \setP_{\cal C} \setminus \{o\}$). By Corollary~\ref{cor:SCCwithsamepayoff} for $G'$ and by induction hypothesis for $G''$, we will get two uniform weak SPEs that can be merged to get a uniform weak SPE for $G$.

%When they are several payoffs in $\setP_{\cal C}$ (with only one layer), then the players can be merged into two meta-players by Lemma~\ref{lem:killer}, such that meta-player~$1$ (resp. meta-player~$2$) has preference relation $p_1 \prec_1 p_2 \prec_1 \ldots \prec_1 p_n$ (resp. $p_n \prec_2 p_{n-1} \prec_2 \ldots \prec_2 p_1$). We will show in the proof how to decompose $G$ into two subgames $G'$ and $G''$ such that the bottom strongly connected component of $G'$ (resp. $G''$) are those components $C \in \cal C$ of $G$ such that $p_C = p_n$ (resp. $p_C \neq p_n$). By Corollary~\ref{cor:SCCwithsamepayoff} for $G'$ and by induction hypothesis for $G''$, we will get two uniform weak SPEs that can be merged to get a uniform weak SPE for $G$.

\begin{proof}[of Lemma~\ref{lem:onelayer}]
The proof is by induction on $|\setP_{\cal C}|$. We solve the basic case $|\setP_{\cal C}| = 1$ by Corollary~\ref{cor:SCCwithsamepayoff}. Suppose that $|\setP_{\cal C}| = n > 1$. By Lemma~\ref{lem:killer}, we have $\prec_i = \prec_{i'}$ or $\prec_i = \prec_{i'}^{-1}$ for all $i, i' \in \Pi$. We can thus merge the players into two \emph{meta-players} ${\cal P}_1$ and ${\cal P}_2$ with their respective preference relations $\prec_1$, $\prec_2$ on $\setP_{\cal C}$ satisfying $o_1 \prec_1 o_2 \prec_1 \ldots \prec_1 o_n$ and $o_n \prec_2 o_{n-1} \prec_2 \ldots \prec_2 o_1$. Notice that ${\cal P}_2$ could not exist.

For the sequel, we need the classical concept of \emph{attractor} of $U \subseteq V$ for ${\cal P}_1$~\cite{2001automata}: it is the set $\Attr_1(U)$ composed of all $v \in V$ from which ${\cal P}_1$ can force, against ${\cal P}_2$, to reach $U$. More precisely, $\Attr_1(U)$ is constructed by induction as follows:
$\Attr_1(U) =  \cup_{k \geq 0} X_k$ such that
\begin{eqnarray*} \label{eq:attr}
X_0 &=& U, \\
X_{k+1} &=& X_k \cup \{v \in V \mid v \mbox{ is controlled by ${\cal P}_1$ and }  \exists (v,v') \in E, v' \in X_k \} \\
&& ~~~~\cup \{v \in V \mid v \mbox{ is controlled by ${\cal P}_2$ and } \forall (v,v') \in E, v' \in X_k \}.
\end{eqnarray*}

Let $\BotSCC' = \{C \in \BotSCC \mid o_C = o_n\}$ and $\BotSCC'' = \BotSCC \setminus \BotSCC'$. We construct a subset $V'$ of $V$ as follows:
\begin{enumerate}
\item Initially $V' \gets \cup \{C \mid C \in \BotSCC'\}$
\item $V' \gets \Attr_1(V')$. Let $\cal D$ be the set of bottom strongly connected components of $\Sub{G}{V \setminus V'}$
\item If $\cal D$ contains components not in $\BotSCC''$, then add all of them to $V'$ and goto~2, else stop
\end{enumerate}
At the end of the process, we get two sets $V'$ and $V'' = V \setminus V'$, and the related subgames $G'$ and $G''$ respectively induced by $V'$ and $V''$.

%Let $\BotSCC' = \{C \in \BotSCC \mid p_C = p_n\}$ and $\BotSCC'' = \BotSCC \setminus \BotSCC'$. We construct a subset $V'$ of $V$ as follows: \textcolor{blue}{Arno : it is clearer ?}
%\begin{enumerate}
%\item Initially $V' \gets \cup \{C \mid C \in \BotSCC'\}$
%\item $W \gets \Attr_1(V')$. Let $\cal D$ be the set of bottom strongly connected components of $\Sub{G}{V \setminus W}$
%\item If $\cal D$ contains components not in $\BotSCC''$, then
%\begin{enumerate}
%\item $T \gets W \cup \cup \{D \in {\cal D} \mid D \not\in {\cal C}\}$
%\item goto~2 with $V' = T$
%\end{enumerate}
%\end{enumerate}

Let us prove by induction on the three steps that (*) for all $v \in V'$, there is a path from $v$ to some $C \in \BotSCC'$. To this end, we denote $W = \Attr_1(V')$ at step~2 and $T = W \bigcup \cup\{D \in {\cal D} \mid D \not\in {\cal C}\}$ at step~3. After step~1, (*) is true (with the empty path from $v$ to $v$). It is also the case after step~2, since by definition of the attractor, there is a path from $v \in W = \Attr_1(V')$ to some $v' \in V'$ for which there is a path to some $C \in \BotSCC'$ by induction hypothesis. Consider now $v \in D$ such that $D \in \cal D$ is added to $W$ in step 3. As $D$ does not belong to $\BotSCC''$ and $D$ is a bottom component of $\Sub{G}{V \setminus W}$, then there must exist a path from $v \in D$ to some $C \in \BotSCC'$ and (*) holds.

By construction each $C \in \BotSCC'$ (resp. $C \in \BotSCC''$) is a bottom strongly connected component of $G'$ (resp. $G''$). Let us prove that neither $G'$ nor $G''$ contain other bottom components. Assume the contrary and let $v$ be a vertex belonging to such a bottom component $D$. By step 3 of the previous process, $v$ cannot belong to $V''$. By (*), $v$ cannot belong to $V'$.
%Let us prove by induction on the three steps (thus on the construction of $V'$) that $v$ cannot belong to $V'$. It is initially the case after step 1. It is also the case after step 2, since by definition of the attractor there is a path from $u \in \Attr_1(V')$ to $V'$ and thus $v \not\in \Attr_1(V')$. When a strongly connected component is added to $V'$ by step 3, as it does not belong to $\BotSCC$ and it is a bottom component of $\Sub{G}{V \setminus V'}$, then there must exist a path from each of its vertices to some $C \in \BotSCC'$, showing that $v \not\in V'$ after step 3.
Therefore the set of bottom strongly connected components of $G'$ and $G''$ is equal to $\BotSCC$.

%Let $\BotSCC' = \{C \in \BotSCC \mid p_C = p_n\}$ and $\BotSCC'' = \BotSCC \setminus \BotSCC'$. We construct a subset $V'$ of $V$ as follows:
%\begin{enumerate}
%\item Initially $V' \gets \cup \{C \mid C \in \BotSCC'\}$
%\item $V' \gets \Attr_1(V')$
%\item If $\Sub{G}{V \setminus V'}$ contains bottom strongly connected components not in $\BotSCC''$, add all of them to $V'$ and goto 2.
%\end{enumerate}
%At the end of the process, we get two sets $V'$ and $V'' = V \setminus V'$, and the related subgames $G'$ and $G''$ respectively induced by $V'$ and $V''$.
%
%By construction each $C \in \BotSCC'$ (resp. $C \in \BotSCC''$) is a bottom strongly connected component of $G'$ (resp. $G''$). Let us prove that neither $G'$ nor $G''$ contain other bottom components. Assume the contrary and let $v$ be a vertex belonging to such a bottom component $D$. By step 3 of the previous process, $v$ cannot belong to $V''$. Let us prove by induction on the three steps (thus on the construction of $V'$) that $v$ cannot belong to $V'$. It is initially the case after step 1. It is also the case after step 2, since by definition of the attractor there is a path from $u \in \Attr_1(V')$ to $V'$ and thus $v \not\in \Attr_1(V')$. When a strongly connected component is added to $V'$ by step 3, as it does not belong to $\BotSCC$ and it is a bottom component of $\Sub{G}{V \setminus V'}$, then there must exist a path from each of its vertices to some $C \in \BotSCC'$, showing that $v \not\in V'$ after step 3. Therefore the set of bottom strongly connected components of $G'$ and $G''$ is equal to $\BotSCC$.

By Corollary~\ref{cor:SCCwithsamepayoff} for $G'$ and by induction hypothesis for $G''$, there exist two uniform strategy profiles $\bar \sigma'$ and $\bar \sigma''$ respectively on $G'$ and $G''$ such that $\bar \sigma'$ (resp. $\bar \sigma''$) is a weak SPE in each $(G',v')$, $v' \in V'$ (in each $(G'',v'')$, $v'' \in V''$). Moreover $\mu(\out{\bar \sigma'}_{v'}) = o_n$ and $\mu(\out{\bar \sigma''}_{v''}) \in P_{\cal C} \setminus \{o_n\}$. The required uniform strategy profile $\bar \sigma$ on $G$ is built such that $\Sub{\bar \sigma}{V'} = \bar \sigma'$ and $\Sub{\bar \sigma}{V''} = \bar \sigma''$. Let us show that it is a weak SPE in all $(G,v)$, $v \in V$. Consider first a subgame $(\Sub{G}{h},v')$ such that $\out{\Sub{\bar \sigma}{h}}_{v'}$ is a play in $G'$ and a one-shot deviating strategy using an edge $(v',v'')$ with $v' \in V'$ and $v'' \in V''$. By step 2  (i.e. by definition of the attractor), $v'$ belongs to ${\cal P}_1$ who has no incentive to use $(v',v'')$ since the deviating play goes to $G''$ for which ${\cal P}_1$ receives an outcome $o_m$ such that $o_m \prec_1 o_n$. Consider next a subgame $(\Sub{G}{h},v'')$ such that $\out{\Sub{\bar \sigma}{h}}_{v''}$ is a play in $G''$ and a one-shot deviating strategy using an edge $(v'',v')$ with $v' \in V'$ and $v'' \in V''$. By step 2, $v''$ now belongs to ${\cal P}_2$ who has no incentive to use $(v'',v')$ since he will receive an outcome $o_m$ such that $o_n \prec_2 o_m$.
\qed\end{proof}

We can now proceed to the proof of Theorem~\ref{thm:uniform}, which is by induction on the number of layers of $\setP$. The case of one layer is treated in Lemma~\ref{lem:onelayer}. In case of several layers, we show in the proof how to decompose $G$ into two subgames $G'$ and $G''$ such that there is only one layer in $G'$ and less layers in $G''$ than in $G$. From the two uniform weak SPEs obtained for $G'$ by Lemma~\ref{lem:onelayer} and for $G''$ by induction hypothesis, we construct the required uniform weak SPE for $G$.

\begin{proof}[of Theorem~\ref{thm:uniform}]
We will prove the theorem by induction on the number of layers and additionally show that for all $v \in V$, $\mu(\out{\bar \sigma}_v) = o_C$ for some $C \in \BotSCC$. Let $\setP' \subseteq \setP_{\cal C}$ be the highest layer of $\setP_{\cal C}$ (with respect to the linear order $<$ over $\setK$). %For this proof, recall that the outcome function is prefix-independent.

If $\setP' = \setP_{\cal C}$, then there is only one layer and the required uniform strategy profile follows from Lemma~\ref{lem:onelayer}.

If $\setP' \subset \setP_{\cal C}$, we define $V' \subset V$ composed of all vertices $v$ for which there exists a path from $v$ to some component $C \in \cal C$ such that $o_C \in \setP'$ (in particular $V'$ includes all such components), and we let $V'' = V \setminus V'$. We obtain two subgames $G'$ and $G''$ respectively induced by $V'$ and $V''$. By construction of $V'$, one easily checks that the union of the bottom strongly connected components of $G'$ and $G''$ is equal to $\BotSCC$. Hence, $G'$ has only one layer (equal to $\setP'$) and $G''$ has one layer less than $G$. It follows (by Lemma~\ref{lem:onelayer} and by induction hypothesis) the existence of two strategy profiles $\bar \sigma'$ and $\bar \sigma''$ respectively on $G'$ and $G''$: $\bar \sigma'$ is a uniform weak SPE in each $(G',v')$, $v' \in V'$, such that $\mu(\out{\bar \sigma'}_{v'}) \in \setP'$, and $\bar \sigma''$ is a uniform weak SPE in each $(G'',v'')$, $v'' \in V''$, such that $\mu(\out{\bar \sigma''}_{v''}) \in \setP \setminus \setP'$. The required strategy profile $\bar \sigma$ on $G$ is built such that $\Sub{\bar \sigma}{V'} = \bar \sigma'$ and $\Sub{\bar \sigma}{V''} = \bar \sigma''$. As in the proof of Lemma~\ref{lem:onelayer}, we consider crossing edges between $G'$ and $G''$. By construction, there is no edge $(v'',v')$ with $v' \in V'$ and $v'' \in V''$ showing that a play starting in $G''$ remains in $G''$. On the contrary, there exist edges $(v',v'')$ with $v' \in V'$ and $v'' \in V''$, but no player has an incentive to use them in a one-shot deviating strategy since the resulting outcome is in a layer smaller than $\setP'$. Therefore, $\bar \sigma$ is a weak SPE in each $(G,v)$.
\qed\end{proof}

\section{A counterexample for countably many players and outcomes}
\label{sec:counterexample}
We proceed to give an example of a game without weak SPE. It shows that the requirement of only finitely many leaf-outcomes is not dispensable in Theorem~\ref{thm:generalgraph} or Theorem~\ref{thm:infinitetree}. In \cite[Section 4.3]{Flesch10} there is an example of a game in extensive form with countably many players, uncountably many outcomes, preference heights $3$, but without weak SPE. Our example is similar, but with only countably many outcomes, one single proper infinite play (i.e.~not ending in a leaf), and preferences of height $3$.

\begin{example}
We consider the initialized game $(G,v_0)$ of Figure~\ref{fig:counterex}. The set of players is $\mathbb N$. The player $i$ acts at most once, at the vertex $v_i$, and can either enter the leaf $l_i$ or move onwards to $v_{i+1}$. The play starts with player $0$ at $v_0$. The outcome attached to reaching $l_i$ is denoted by $2^{i}1^\omega$, the outcome attached to the infinite path $v_0v_1v_2\ldots$ is denoted by $0^\omega$. The preferences of player $i$ are given by $p \prec_i q$ iff $p(i) < q(i)$.

The game $(G,v_0)$ has no SPE. To prove this statement, it is enough to show that there is no very weak SPE by Proposition~\ref{prop:weak-veryweak} and since every player only acts one. Assume by contradiction that there exists a very weak SPE $\bar \sigma$. In each subgame $(\Sub{G}{h},v_i)$, the play induced by $\bar\sigma$ cannot be the one with outcome $0^\omega$. Otherwise player $i$ has a profitable one-shot deviating strategy by moving to leaf $l_i$ (by increasing his payoff from 0 to 1). Therefore, for all $k$, there exists a player $i \geq k$ who moves to leaf $l_{i}$. Let $i$ be the first such player. It follows that in $(\Sub{G}{h},v_i)$, he can increase his payoff from 1 to 2 by moving to $v_{i+1}$ instead to $l_i$, contradiction.

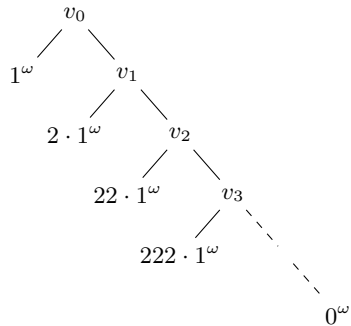
\begin{figure}
\centering
\begin{tikzpicture}[level distance=8mm]
\node{$v_0$}[sibling distance=14mm]
	child{node{$1^\omega$}[sibling distance=10mm]}
	child{node{$v_1$}[sibling distance=14mm]
		child{node{$2\cdot 1^\omega$}}
		child{node{$v_2$} [sibling distance=14mm] edge from parent[solid]
			child{node{$22\cdot 1^\omega$} edge from parent[solid]}
			child{node{$v_3$} [sibling distance=14mm] edge from parent[solid]
				child{node{$222\cdot 1^\omega$} edge from parent[solid]}
					child{node{} edge from parent[dashed]
						child{node{}edge from parent[draw=none]}
						child{node{$0^\omega$} edge from parent[dashed]}
					}	
			}
		}	
	};
\end{tikzpicture}

 \caption{A game with no weak SPE}
 \label{fig:counterex}
 \end{figure}

\end{example}

%Let the set of players be the set of natural numbers $\mathbb{N}$. Each player $n$ makes one decision once, at stage $n$. More specifically, at the root, player $0$ chooses to stop or to continue, and for all $n \in \mathbb{N}$, if all players less than $n$ have chosen to continue, player $n$ can choose to stop or to continue. If player $n$ makes the play stop, then all players less than $n$ receive payoff $2$, and the others receive $1$. If all players choose to continue, they all receive payoff $0$.

%Since every player only acts once, weak SPE and SPE coincide. Let us show by contradiction that there is no (weak) SPE in the game. Every player can avoid payoff $0$ by stopping. Therefore the SPE has no infinite play, i.e. for all $n$ there exists a player $n+k$ who stops. Let $n$ be the first player who stops. By changing strategies, the player can increase her payoff from $1$ to $2$, contradiction.

\bibliographystyle{plain}
\bibliography{biblio}

\end{document}